\documentclass[prx]{revtex4-2}
\usepackage[paperwidth=210mm,paperheight=297mm,centering,hmargin=2.5cm,vmargin=2.5cm]{geometry}

\usepackage{amsmath, amssymb, natbib, graphicx, algorithm2e, amsthm}
\newtheorem{theorem}{Theorem}[section]
\newtheorem{proposition}[theorem]{Proposition}
\newtheorem{assumption}[theorem]{Assumption}

\def\rhob{\boldsymbol{\rho}}

\def\thetab{\boldsymbol{\theta}}

\def\ab{\boldsymbol{a}}
\def\bb{\boldsymbol{b}}

\def\xb{\boldsymbol{x}}

\def\zb{\boldsymbol{z}}

\def\Gb{\boldsymbol{G}}
\def\Wb{\boldsymbol{W}}
\def\Xb{\boldsymbol{X}}

\def\grad{\nabla}

\def\lap{\Delta}

\def\lap{\Delta}

\def\RR{\mathbb{R}} \def\NN{\mathbb{N}} 
\def\EE{\mathbb{E}}

\def\<{\langle} \def\>{\rangle}

\DeclareRobustCommand{\tr}{\operatorname*{tr}}

\usepackage[html]{xcolor}
\definecolor{green}{HTML}{0f6852}
\definecolor{lightblue}{HTML}{01abe9}

\usepackage{hyperref}
\hypersetup{
  colorlinks = true,
  linkcolor = green,
  citecolor = lightblue,
  urlcolor = lightblue
}

\begin{document}

\title{Active Importance Sampling for Variational Objectives Dominated by Rare Events: Consequences for Optimization and Generalization}

\author{Grant M. Rotskoff}
\affiliation{Dept. of Chemistry, Stanford University, Stanford, CA 94305}
\email{rotskoff@stanford.edu}
\author{Andrew R. Mitchell}
\affiliation{Dept. of Chemistry, Stanford University, Stanford, CA 94305}
\author{Eric Vanden-Eijnden}
\affiliation{Courant Institute, New York University, New York, NY 10012}
\email{eve2@cims.nyu.edu}

\begin{abstract}
 Deep neural networks, when optimized with sufficient data, provide accurate representations of high-dimensional functions; in contrast, function approximation techniques that have predominated in scientific computing do not scale well with dimensionality.
 As a result, many high-dimensional sampling and approximation problems once thought intractable are being revisited through the lens of machine learning.
 While the promise of unparalleled accuracy may suggest a renaissance for applications that require parameterizing representations of complex systems, in many applications gathering sufficient data to develop such a representation remains a significant challenge.
 Here we introduce an approach that combines rare events sampling techniques with neural network optimization to optimize objective functions that are dominated by rare events.
 We show that importance sampling reduces the asymptotic variance of the solution to a learning problem, suggesting benefits for generalization.
 We study our algorithm in the context of solving high-dimensional PDEs that admit a variational formulation, a problem with applications in statistical physics and implications in machine learning theory.
 Our numerical experiments demonstrate that we can successfully learn even with the compounding difficulties of high-dimension and rare data.
\end{abstract}

\maketitle

% \begin{keywords}%
%   Partial Differential Equations; Importance Sampling; Rare Events; Backward Kolmogorov Equation; Variational Monte Carlo%
% \end{keywords}

\section{Introduction}

Deep neural networks (DNNs) have become an essential tool for a diverse set of problems in data science and, increasingly, the physical sciences~\cite{carleo_machine_2019}.
The uncommonly robust approximation properties of DNNs undergird the successes of deep learning in seemingly disparate problems~\cite{lecun_deep_2015}. 
The power of approaches based on deep learning is evident in high-dimensional settings where most classical tools from numerical analysis break down, due to the curse of dimensionality~\cite{donoho1989}.
Many compelling questions in statistical physics require precise knowledge of high-dimensional functions, objects which can be challenging to represent and compute, suggesting that machine learning may have a transformative role to play.

Of course, challenges arise when using machine learning techniques in the physical sciences that do not appear in conventional settings.
Unlike in computer vision and natural language processing, curated data sets are not typically available for physical problems that we intend to solve \emph{de novo}. 
As a result, we must generate the data either experimentally or computationally that we use to train our models.

Of particular interest in this context are problems involving high-dimensional partial differential equations (PDE) that can be formulated as variational minimization problems. Many PDEs of interest in statistical mechanics and quantum mechanics admit such a variational principle, and they lend themselves naturally to solution by machine learning techniques since the objective function can serve as a loss to train a neural network used to represent the solution. How to generate data to evaluate this objective constitutes, perhaps, the core challenge in problems of this type because the data that dominates the objective may be rare if sampled naively. In this work, we address this sampling problem.

\textbf{Neural networks for variational PDEs.}---Consider a PDE whose  solution can be found via the minimization problem
\begin{equation}
    \min_{f\in \mathcal{F}} \mathcal{I}(f)
    \label{eq:ritz1}
\end{equation}
Here 
\begin{equation}
    \mathcal{I}(f) = \int_\Omega \mathcal{L}(\xb,f) d\nu(\xb),
    \label{eq:ritz12}
\end{equation}
where $\Omega \subset \RR^d$, $\nu$ is some positive measure, and $\mathcal{L}(\xb,f)$ is some Lagrangian depending on $\xb$ as well as $f$ and its derivatives: Typical examples are
\begin{equation}
    \mathcal{L}(\xb,f) = \tfrac12 |\grad f(\xb)|^2 + V(\xb) |f(\xb)|^2, \qquad d\nu(\xb) = d\xb,
    \label{eq:ritz2}
\end{equation}
where $V:\Omega \to \RR$ is some potential, which gives the time-independent Schr\"odinger equation if we impose $\int_\Omega |f(\xb)|^2 d\xb = 1$, or 
\begin{equation}
    \mathcal{L}(\xb,f) = \tfrac12 |\grad f(\xb)|^2, \qquad d\nu(\xb) = e^{-\beta V(\xb)}d\xb \qquad (\beta>0)
    \label{eq:ritz2b}
\end{equation}
which gives the time-independent backward Kolomogorov equation if we impose some boundary conditions.

Variational Monte Carlo (VMC) procedures~\cite{hoggan_chapter_2016} have been used to compute solutions to PDEs that admit this formulation. In this context, solutions are often computed using the Ritz method, which essentially amounts to optimizing the weights of specified, hand-chosen basis elements.  Methods based on neural networks~\cite{eigel_variational_2019,e_deep_2017} offer an alternative to VMC which may need less \textit{a~priori} information about the solution by relying on the approximation power of these networks. 

\begin{figure}
 \centering
 \includegraphics[width=0.48\linewidth]{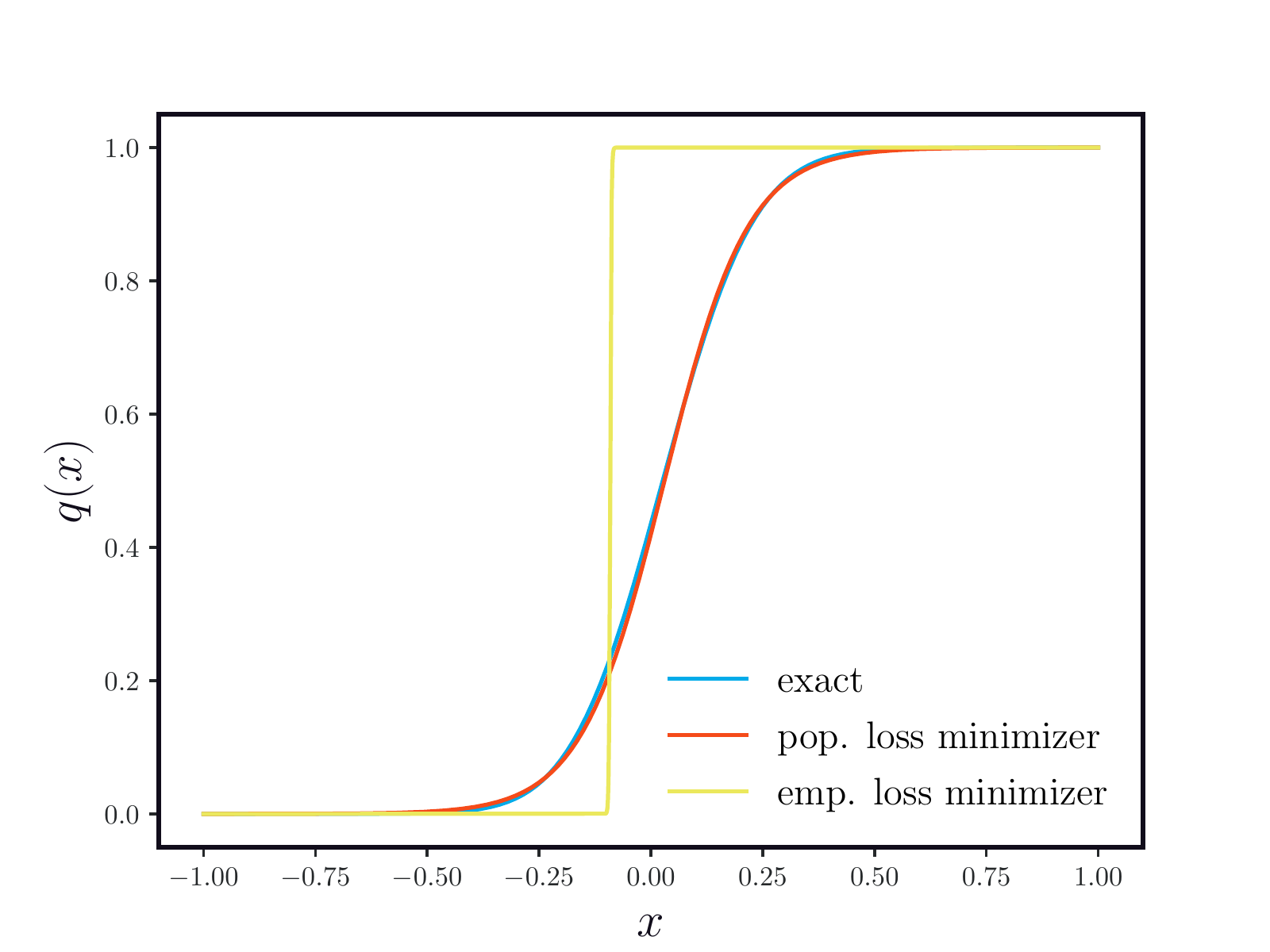}
 \includegraphics[width=0.48\linewidth]{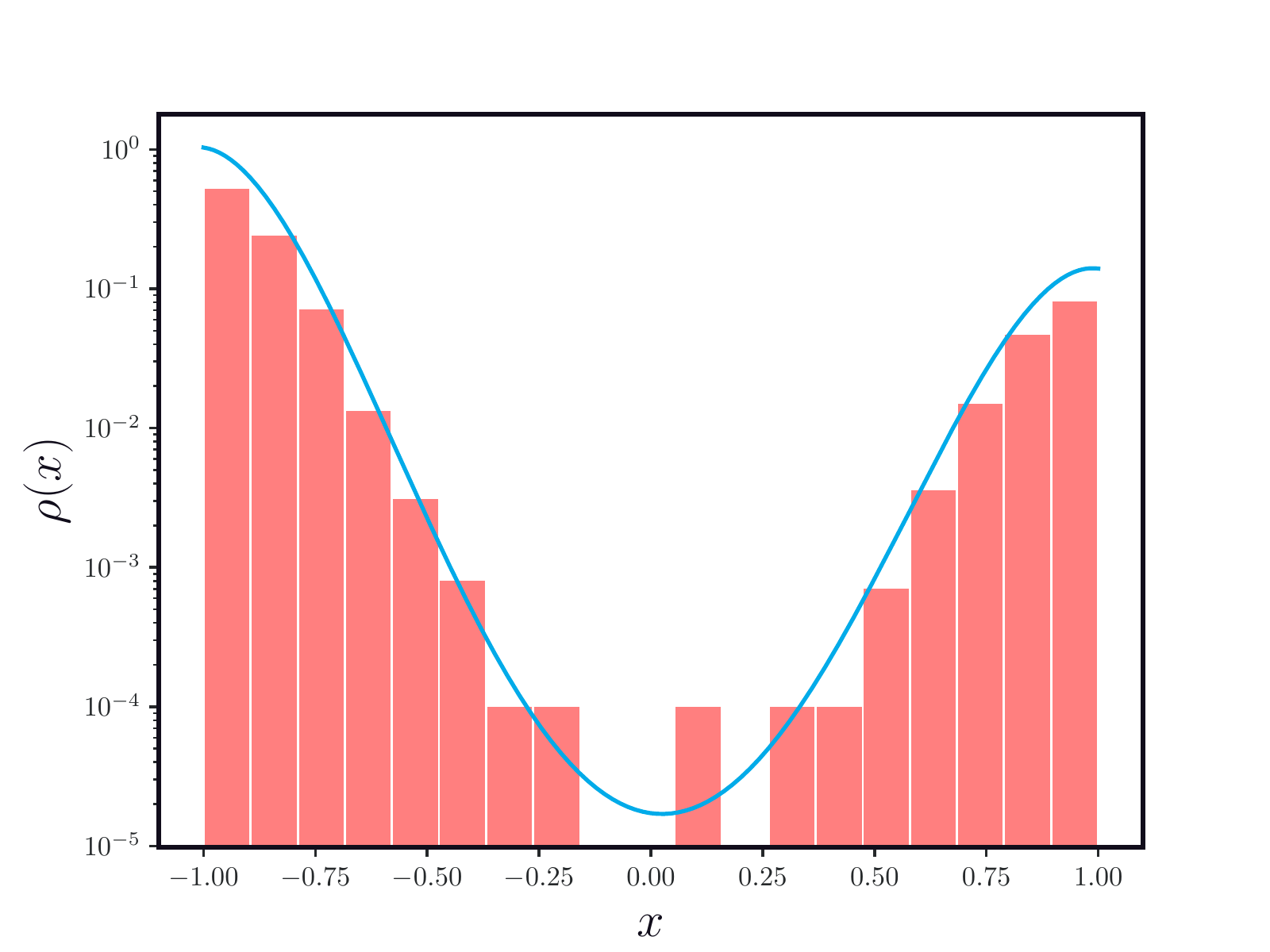}
 \includegraphics[width=0.48\linewidth]{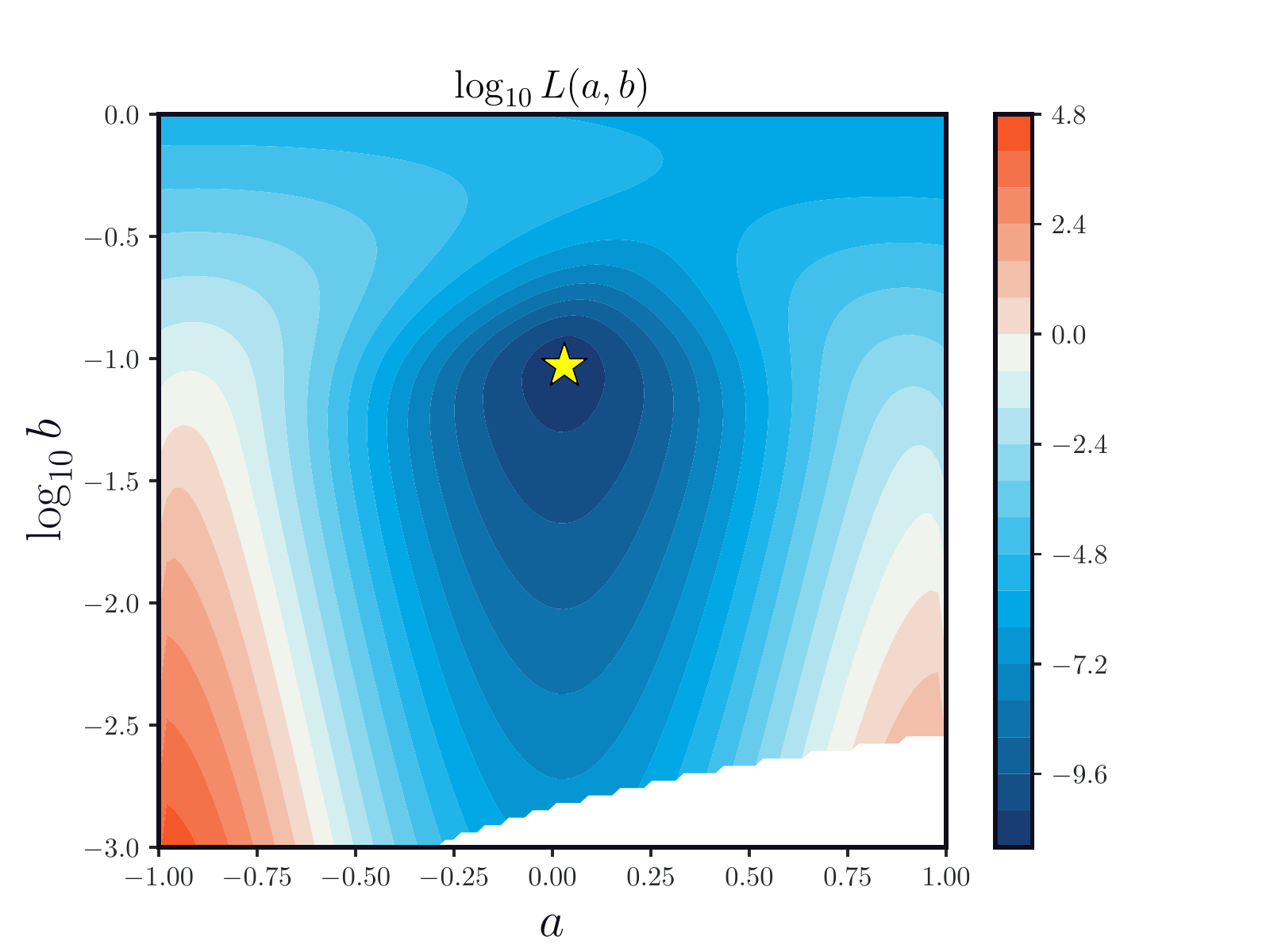}
 \includegraphics[width=0.48\linewidth]{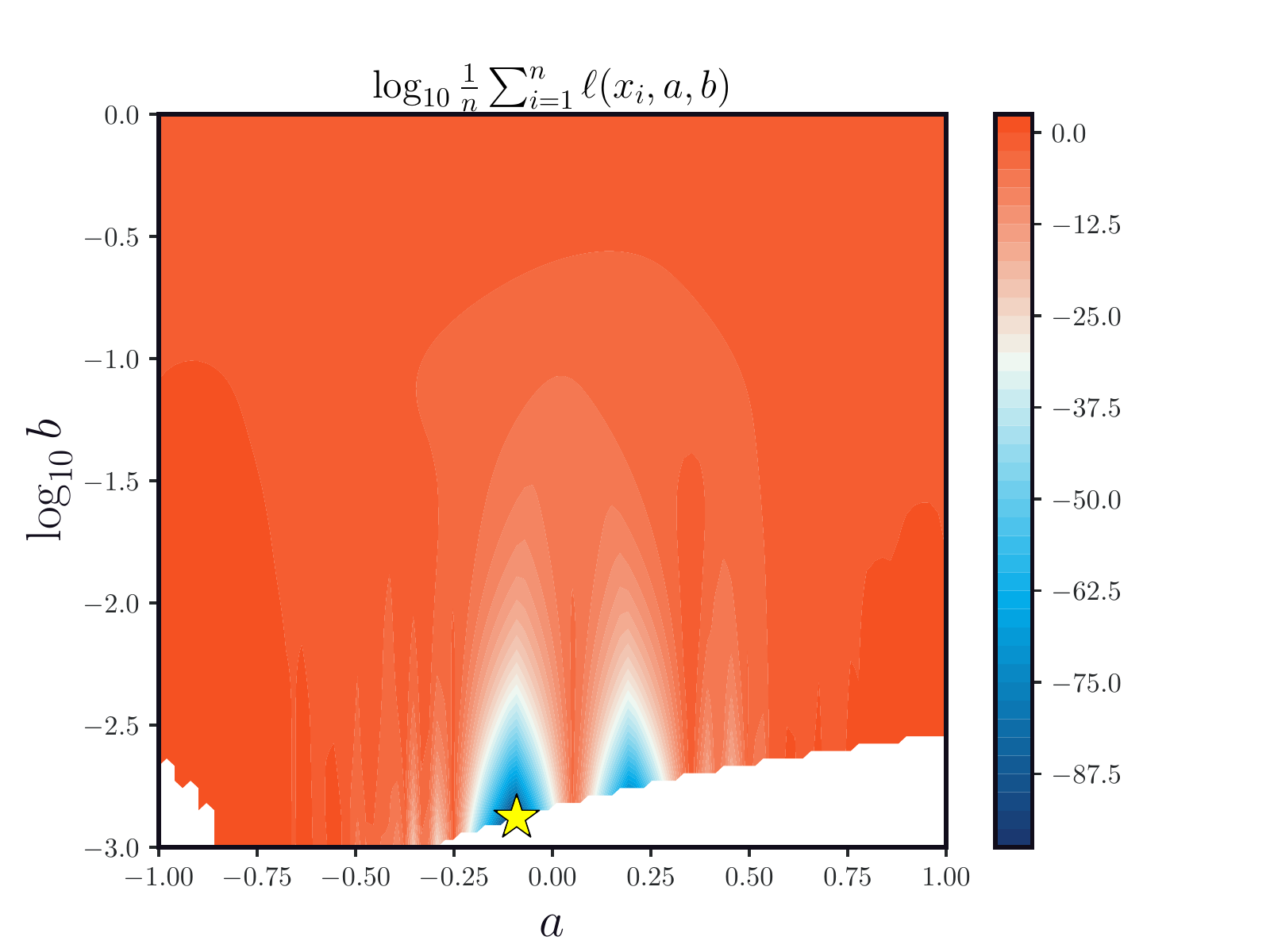}
 \caption{
 Simple illustrative example where the objective $I(q)= \int_{x_1}^{x_2} |q'(x)|^2 e^{-\beta V(x)} dx$, with $V(x) = (1-x^2)^2 + x/10$, $\beta = 10$, and $x_1$, $x_2$ at the minima of $V(x)$, is minimized over sigmoid functions with two parameters $a$ and $b$, controlling respectively their location and width -- in this example, functions of this type do a good job at capturing the minimizer if $a$ and $b$ are properly are adjusted. Top left panel: the minimizers of the population and the empirical losses are compared to the actual solution; top right panel: the histogram of the data acquired by drawing $10^4$ independent samples  from the Gibbs distribution restricted on $x\in[x_1,x_2]$ is compared to the exact density; bottom panels: the population loss (left) is compared to the empirical loss (right) estimated on the data. Because the data is somewhat sparse near the maximum of $V(x)$ that dominates the objective, the empirical loss does a bad job at capturing the features of the population loss -- note in particular the different scale and the added ruggedness in the empirical loss. As a result, the function optimized over this empirical loss differs significantly from the minimizer of the population loss that approximates the exact solution well, leading to a large generalization error. Note that here we sampled the measure and identified the minimizers by brute-force: in more complicated situations there is the added difficulty of performing this sampling, and minimizing the empirical loss by SGD.  For further details on this example, see Appendix~\ref{sec:one:d}.
 }
 \label{fig:losscomm1d}
\end{figure}

\textbf{Data acquisition and importance sampling.}---Training a neural network to represent the solution of the PDE by minimizing~\eqref{eq:ritz1} requires estimating the integral~\eqref{eq:ritz12}. 
Because there is no data set given beforehand, the most straightforward implementation samples data points on~$\Omega$ from the measure~$\nu$ properly normalized. 
While natural, this approach is by no means optimal and it could even fail if the expectation of $\mathcal{L}(\xb,f)$ is dominated by events that are rare on $\nu$: a simple  example illustrating this point is shown in Fig.~\ref{fig:losscomm1d}. If $\Omega$ is high dimensional, the variance of a simple estimator using unbiased samples from $\nu$ will typically be large compared to its mean squared, and some form of importance sampling will therefore be required.
If we were interested in estimating the loss $I(f)$, it is well-known that the optimal way to draw samples would be to use the reweighted measure $d\tilde \nu(\xb) = I^{-1}(f) |\mathcal{L}(\xb,f)| d\nu(\xb)$ and reweight the samples consistently using $I(f) |\mathcal{L}(\xb,f)|^{-1}$ ~\cite{awad_zero-variance_2013}. 
The difficulties with this approach are that the reweighted measure $\tilde \nu$ may not be easy to sample, and the reweighting factor involves the unknown value $I(f)$. 

We show below that an importance sampling strategy can, however, be applied to reduce the variance of the estimator for the loss (as our training procedure relies on data generated at every training step) associated with variational problems of the type \eqref{eq:ritz1}.
These methods are widely used in applications from statistical mechanics because offer a remedy to the problem of an objective dominated by rare data, but they are often rendered intractable by the need for precise knowledge about where and how to sample to design an importance sampling scheme. 
In our context, however,  we can use the current estimate of the solution to inform and enhance the sampling in regions of the domain that contribute to the objective. 
The efficiency of such \emph{active} learning approaches will be demonstrated below.

\textbf{Reactive events and committor.}---As a specific application of practical interest that illustrates the general issues outlined above, we will focus on optimizing an objective of type~\eqref{eq:ritz2b} for a target function known as the ``committor function,'' or committor in short. 
The committor is useful to identify reaction pathways and sample reactive trajectories in problems displaying metastability, a central question in statistical mechanics with decades of work behind it. 
In this context, the committor describes the probability that a configuration will ``react'', by transiting from one metastable basin to another under the stochastic dynamics of the system under consideration.
Parameterizing the committor accurately (as many functions related to rare transitions in applications in condensed matter physics) requires samples from configurations that are rare under the Boltzmann distribution, a fact emphasized by the ubiquity of importance sampling methods for free energy calculations.
With this in mind, calculating the committor epitomizes why a naive sampling strategy will not succeed in general and importance sampling is necessary.

\textbf{Related works.}---Importance sampling and other variance reduction techniques have appeared in a variety of contexts in machine learning.
\citet{csiba_importance_2018} described and analyzed an algorithm that does importance sampling of the training set to adaptively select minibatches and accelerate gradient descent.
Their work formalizes an approach, represented in a large body of work \cite{nesterov_efficiency_2012,NIPS2012_4633, NIPS2013_4937}, that aims to reduce the variance in the gradients when optimizing using stochastic gradient descent.
In a separate line of inquiry, \citet{fan_importance_2010} uses importance sampling to perform approximate Bayesian inference in continuous time Bayesian networks. 
Our setting differs substantially from these works, as we are principally concerned with problems in which the data set is sampled on-the-fly from a Boltzmann distribution.
Furthermore, we require importance sampling for the learning to be tractable at all, whereas the aforementioned works seek to accelerate optimization in otherwise tractable learning problems.
Our theoretical results suggest that these previously studied approaches benefit generalization.

Our work parallels a line of inquiry in the Quantum Monte Carlo literature which has demonstrated the utility of neural network ansatze for electronic structure problems~\cite{han_solving_2019, hermann_deep-neural-network_2020, pfau_ab_2020}.
Though the physical setting is quite different from the one we consider here, these works also rely on a strategy in which the data is collected online and there is feedback between training and data collection.
Regarding the application to metastability, transition path sampling methods are perhaps the most closely related to our approach~\cite{bolhuis_transition_2002, maragliano_string_2006}.
Our applications are heavily influenced by the perspective of potential theory~\cite{bovier_metastability_2002} and transition path theory~\cite{e_towards_2006,e_transition-path_2010}, which use the notion of the committor function (discussed in detail below) to characterize metastability.
\citet{khoo_solving_2018} first considered the problem of learning committor functions from the perspective of solving high-dimensional PDEs but did not address the sampling issues that can arise in computing the objective.
Our approach most closely follows that of \citet{li_computing_2019}, who also examined the problem of optimizing a representation of the committor using neural networks on low-dimensional landscapes.
Our work extends this approach in several important ways: first, our algorithm uses an \emph{active} approach---the importance sampling directly uses the committor function meaning that there is feedback between the optimization and the data collection.
In high-dimensional systems in which selecting a reaction coordinate presents a challenging design problem, our approach is crucial for effective sampling because we avoid explicitly constructing a reaction coordinate.

\textbf{Main contributions.}---First, under very general assumptions, we show that importance sampling asymptotically improves the generalization error.
Next, we describe an algorithm for \emph{active} importance sampling that enables variance reduction for the estimator of the loss function, even in high-dimensional settings.
Finally, we demonstrate numerically that this algorithm performs well both on low and high-dimensional examples and that, even when the total amount of data is fixed, optimizing the variational objective fails when importance sampling is not used.

\section{Online learning and generalization error}
\label{sec:online:gen}

Suppose we parametrize the function $f$ entering the objective function in~\eqref{eq:ritz1} using a neural network, i.e. we set $f(\xb) = f(\xb,\thetab)$, where $f(\cdot,\thetab)$ is the network output and $\thetab\in \RR^N$ collectively denotes all the parameters entering this network. This turns~\eqref{eq:ritz1} into an objective function for the parameters~$\thetab$:
\begin{equation}
    \label{eq:poprisk}
    L(\thetab) = I(f(\cdot,\thetab)) = \int_\Omega \ell(\xb,\thetab) d\nu(\xb)
\end{equation}
where
\begin{equation}
    \label{eq:lagtheta}
    \ell(\xb,\thetab) = \mathcal{L}(\xb,f(\cdot,\thetab))
\end{equation}
In the jargon of machine learning, $L(\thetab)$ is called the population loss or risk, and in practice, it must be estimated using an empirical estimate. 
The simplest choice for the empirical loss is 
\begin{equation}
    \label{eq:emprisk}
    L_n(\thetab) = \frac1n \sum_{i=1}^n \ell(\xb_i,\thetab) 
\end{equation}
where $\{\xb_i\}_{i=1}^n$ are \textit{iid} drawn from $\nu$. This offers the possibility to optimize the parameters $\thetab$ by gradient descent (GD), i.e. using
\begin{equation}
    \label{eq:sgd}
    \thetab^{k+1} = \thetab^k - \alpha \nabla_{\thetab} L_n(\thetab^k), \qquad k=0,1, 2, \ldots
\end{equation}
where $\thetab^k$ denote the successive updates of the parameters starting from some initial $\thetab^0$ and $\alpha>0$ is some time step (learning rate). In situations in which no data set is available beforehand, it is customary to use online learning, i.e. to generate new independent batches of data $\{\xb_i\}_{i=1}^n$ after each (or a few) step(s) of the GD update. In~\eqref{eq:sgd}: each $\{\xb_i\}_{i=1}^n$ is called a minibatch, and the update in~\eqref{eq:sgd} is the widely used stochastic gradient descent (SGD) algorithm.

In this setup, the main issue becomes how to assess the quality of an approximation of the minimizer(s) of the population risk that we obtain using SGD.
To phrase this question more precisely, let us denote by $\{\bar\thetab^k\}_{k\in \NN_0}$ the successive update of the parameters by GD over the population risk, i.e.
\begin{equation}
    \label{eq:gd}
    \bar \thetab^{k+1} = \bar \thetab^k - \alpha \nabla_{\thetab} L(\bar\thetab^k), \qquad k=0,1, 2, \ldots
\end{equation}
Let us assume that:
\begin{assumption}
\label{as:1}
Given some initial value $\bar\thetab^0$, the GD update in~\eqref{eq:gd} converges towards a local minimizer of the population risk, $\thetab^*= \lim_{k\to\infty} \bar\thetab^k$, satisfying
\begin{equation}
    \label{eq:min:prop}
    \nabla_{\thetab} L(\thetab^*) = 0, \qquad H^* = \nabla_{\thetab} \nabla_{\thetab} L(\thetab^*) \quad \text{is positive-definite}
\end{equation}
\end{assumption}
Note that this assumption does not specify how good the local minimizer $\thetab^*$ is (i.e. how close $L(\thetab^*)$ is from $\min_{\thetab} L(\thetab)$) but it requires that $L(\thetab)$ be strictly convex in the vicinity of $\thetab^*$.
Similar assumptions have been used to study SGD as variational inference~\cite{mandt_stochastic_2017}.
This assumption implies:

\begin{proposition}
\label{th:1a}
The sequence $\{\thetab^k\}_{k\in \NN_0}$ obtained using the SGD update in~\eqref{eq:sgd} starting from $\thetab^0 = \bar\thetab^0$ and using an independent batch of data $\{\xb_i\}_{i=1}^n$ drawn from $\nu$ at every step is such that
\begin{equation}
\label{eq:generr1}
    \lim_{k\to\infty} \lim_{n\to\infty} n \EE_D [L( \thetab^k) - L(\thetab^*)] = \tfrac12\alpha \tr [C^* H^*]
\end{equation} 
Here $\EE_D$ denotes expectation over all the batches used to compute the sequence $\thetab^k$, and $C^*$ is the $N\times N$ tensor that solves
\begin{equation}
    \label{eq:C}
    H^* C^* + C^* H^* - \alpha C^* H^* C^* = B^*
\end{equation}
where $B^*$ is the covariance of $\nabla_{\thetab} \ell(\xb,\thetab^*)$  (using  $\nabla L(\thetab^*) =0$)
\begin{equation}
    \label{eq:B}
    B^* = \int_\Omega \nabla_{\thetab} \ell(\xb,\thetab^*) [\nabla_{\thetab} \ell(\xb,\thetab^*)]^T d\nu(\xb)
\end{equation}
\end{proposition}
The proof of this proposition is given in Appendix~\ref{app:generr}. Essentially, it amounts to linearizing the sequence $\{\thetab^k\}_{k\in \NN_0}$ from SGD around $\{\bar\thetab^k\}_{k\in \NN_0}$ from GD: the resulting sequence is the discretized version of an Ornstein-Ulhenbeck process that can be analyzed exactly.

Even though the statement in~\eqref{eq:generr1} is only asymptotic in $n$ and $k$, it suggests that for large $k$ and large $n$, we will have
\begin{equation}
\label{eq:generr2}
     \EE_D L(\thetab^k) = L(\thetab^*) + \tfrac12 n^{-1} \alpha \tr [C^* H^*] + \text{higher order corrections in $n$}
\end{equation} 
Therefore, if we can guarantee that $\thetab^*$ is a good local minimizer of the loss (which has to do with the choice of network architecture and how well-tailored it is to the problem at hand, the choice of the initial $\bar\thetab^0$, etc.), \eqref{eq:generr1} indicates that the error made by learning using SGD rather than GD can be controlled by: (i) increasing the size $n$ of the batches, (ii) decreasing the learning rate $\alpha$, or (iii) reducing $\tr [C^* H^*]$. The first two observations are standard and are at the core of the Robbins-Monro stochastic approximation procedure~\cite{robbins1951}. The third observation is also not surprising from the proof of Proposition~\ref{th:1a} which shows that  $n^{-1} \alpha C^*$ is asymptotic covariance of the update $\thetab^k$ from the SGD sequence around its mean $\bar\thetab^k$. 

Interestingly, reducing $\tr [C^* H^*]$ essentially amounts to reducing   $\tr B^*$ with $B^*$ defined in~\eqref{eq:B}. Indeed, to leading order in $\alpha$, we see from~\eqref{eq:C} that
\begin{equation}
    \tr[C^* H^*] = \tfrac12 \tr B^* + O(\alpha)
\end{equation}
Reducing $\tr B^*$ is precisely what we show how to do next using importance sampling.

\section{Active learning with umbrella sampling and replica exchange}
\label{sec:active2}

\begin{figure}
 \centering
 \includegraphics[width=\linewidth]{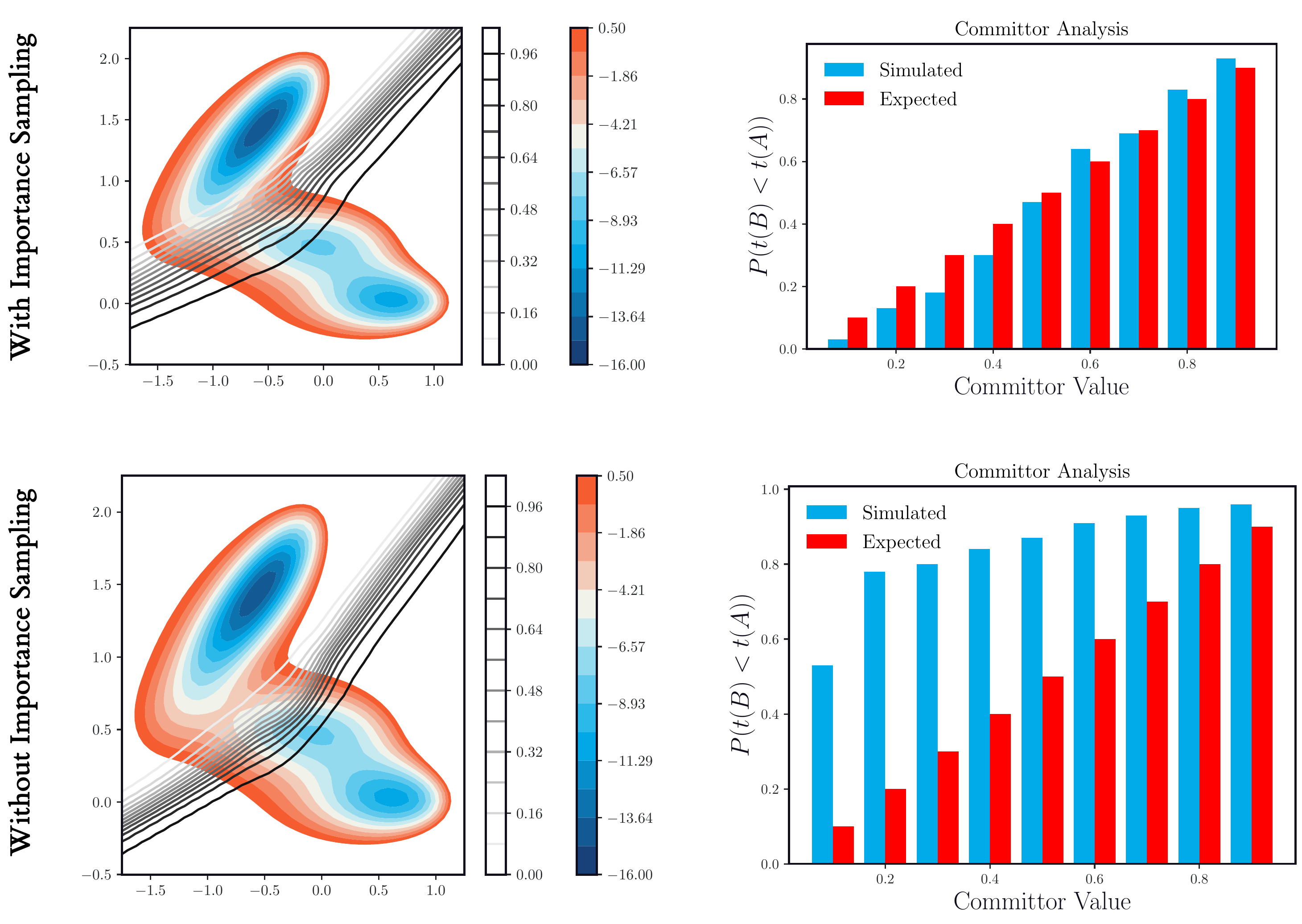}
 
 \caption{
 Simple illustrative experiment; the potential energy function is a 2D mixture of Gaussians and the committor is represented as a single hidden layer neural network. Top (results using our active importance sampling algorithm): Filled contours illustrate the M\"uller-Brown potential \eqref{eq:v_mb} with the isocommittor lines are shown from white to black. Notably, level set $q=0.5$ coincides with the saddle, as expected. Right: We verify that this solution is consistent with the expected committor function by sampling 100 configurations for each window $0.1, 0.2, \dots$ and running Langevin trajectories to compute estimate the committor probability. The fraction of trajectories reaching $B$ before $A$ is close to the expected value. 
 Bottom: Without importance sampling, the optimization converges to a representation of the committor $q$ with poor performance. Left: The contours of the committor fail to localize to the transition region. Right: The committor analysis shows that without importance sampling, the results deviate strongly from the expected probabilities. 
 }
 \label{fig:mb}
\end{figure}

To reduce the variance of the estimator for $L$, we will use an importance
sampling strategy that combines umbrella sampling~\cite{torrie1977} (cf. stratification~\cite{dinner2020}) with replica exchange~\cite{swendsen1986,fukunishi_hamiltonian_2002}. The first method uses windowing functions to enhance the sampling in otherwise rarely sampled regions of the data distribution; the second allows for exchange between these windows to accelerate sampling even further. Both these methods are widely used: the novelty lies in the way we actively define the windowing functions using the current estimate of the target function $f$ by its network representation $f(\cdot, \thetab)$.
In Appendix~\ref{sec:active1} we discuss an active importance sampling scheme based on direct reweighting which could be used to reduce the variance in the estimate of the gradient of the loss. 

Let us denote these windowing functions as a set of non-negative functions
$W_l(\xb)\ge0$ with $l=1,\ldots, L$ such that
\begin{equation}
 \label{eq:2}
 \forall \xb \in \Omega \quad : \quad \sum_{l=1}^L W_l(\xb) = 1,
\end{equation}
Given any test function $\phi:\Omega \to \RR$, we can write
\begin{equation}
 \label{eq:3}
  \EE_\nu \phi = \sum_{l=1}^L \int_{\RR^d} \phi(\xb) W_l(\xb)
  d\nu(\xb)  \equiv \sum_{l=1}^L w_l \, \EE_l \phi 
\end{equation}
where we defined the expectation
\begin{equation}
 \label{eq:4}
 \EE_l\phi  = Z_l^{-1} \int_{\RR^d} \phi(\xb) W_l(\xb)
 d\nu(\xb) \quad \text{where} \quad Z_l = \int_{\RR^d} W_l(\xb)
 d\nu(\xb)
\end{equation}
as well as the weights
\begin{equation}
 \label{eq:5}
 w_l = \EE_\nu W_l.
\end{equation}
% \eqref{eq:3} can be written as
% \begin{equation}
%  \label{eq:3b}
%  \EE_\nu f = \sum_{l=1}^L \< f \>_l w_l
% \end{equation}
By choosing $\phi(\xb) = W_{l'}(\xb)$ in this expression, we deduce
that the weights satisfy the eigenvalue problem \cite{thiede2016eigenvector}
\begin{equation}
 \label{eq:13}
 w_{l'} = \sum_{l=1}^L w_{l} p_{ll'}, \quad l'=1,\ldots, L, \quad
 \text{subject to} \ \ \sum_{l=1}^L w_l = 1,
\end{equation}
where we defined
\begin{equation}
 \label{eq:11}
 p_{ll'} = \< W_{l'} \>_l
\end{equation}
In practice, we can sample
$Z_{l}^{-1} W_l(\xb) d\nu(\xb) $ by Metropolis-Hastings
Monte-Carlo on a potential biased by $- \log W_l(\xb)$ and  compute expectations in this ensemble as
\begin{equation}
 \label{eq:8}
 \EE_l\phi  \approx \frac1n \sum_{i=1}^n  \phi(\xb_{i,l}), \qquad
 \xb_{i,l}\sim Z_{l}^{-1} W_l(\xb) d\nu(\xb)
\end{equation}
This allows us to estimate $\EE_l \phi $ in~\eqref{eq:13} as well as
$p_{ll'}$ in~\eqref{eq:11}: knowledge of the latter quantity enables us to solve the eigenvalue problem in~\eqref{eq:13} to find
the weights $w_l$, and finally estimate~$\EE_\nu \phi$ via~\eqref{eq:3}. 
The sampling can be accelerated by using replica exchange between the ensembles in the different windows, which alleviate potential problems due to metastability within these windows. 
That is, at an interval $t_{\rm swap}$, we attempt to exchange the configuration in window $l$ with a configuration in a neighboring window, accepting the move with a Metropolis acceptance probability.

As of yet, we have not specified the windowing functions $W_l(\xb)$.
Because the $W_l(\xb)$ determine where the samples concentrate a ``good'' choice of these functions is crucial for the success of the sampling scheme. 
Here we propose to make this choice adaptive to the function $f(\xb,\thetab)$ that is being optimized, by dividing space into regions where $f(\xb,\thetab)$ takes specific values. To this end, let
\begin{equation}
 \label{eq:7}
 \sigma(u) = \frac1{1+e^{-u}}
\end{equation}
and given $u_0< u_1< u_2< \cdots < u_L$ and some $k>0$, define
\begin{equation}
 \label{eq:6}
 W_l(\xb) = \sigma\left(k(f(\xb,\thetab)-u_{l-1})\right)
 -\sigma\left(k(f(\xb,\thetab)-u_{l})\right), \qquad l=1,\ldots,L
\end{equation}
In the applications considered below $f$ is a probability and hence its range is restricted to $[0,1]$. As a result we have
\begin{equation}
 \label{eq:9}
 \begin{aligned}
  \forall \xb \in \Omega \ \ : \ \ \sum_{l=1}^L W_l(\xb) & =
  \sigma\left(k(f(\xb,\thetab)-u_{0})\right) -
  \sigma\left(k(f(\xb,\thetab)-u_{L})\right)                  \\   & \ge
  \sigma(-ku_{0})-\sigma(k(1-u_{L}))
 \end{aligned}
\end{equation}
That is if we take $k$ large enough and pick $u_0=-a$ and $u_L = 1+a$ with $a>0$ such that $k a\gg1$, the non-negative functions $W_l(\xb)$ can be made to satisfy~\eqref{eq:2} to arbitrary precision exponentially fast in $ak$. The functions $W_l(\xb)$ are also peaked around $f(\xb,\theta) = \frac12 (u_{l}+u_{l-1})$ which means that by
taking enough values of $u_l$ between $u_0=-a$ and $u_L = 1+a$ we can
cover all the range of possible values for $f(\xb,\thetab)$.
The $u_l$ can be spaced linearly, or, to concentrate sampling near the rapidly varying part of the committor function can be spaced geometrically away from $u_l = 1/2$.

An explicit scheme putting these steps together with SGD is described in Algorithm~\ref{alg:basic}, where we
provide a description of the most straightforward implementation of
our approach.
Algorithm~\ref{alg:basic} is sequential; a version in which we evolve $\xb_{i,l}$ and $\thetab$ concurrently would allow for significant wallclock speed-ups.
\begin{center}
\begin{algorithm}[t]
 \label{alg:basic}
  \KwData{Lagrangian $\ell(\xb,\thetab) = \mathcal{L}(\xb,f(\cdot, \thetab))$, initial $\thetab$, $n\in \NN$, $L\in \NN$, $\alpha>0$, $k>0$, $u_0<\ldots< u_L$.}
  %\State Main routine:
  %\State {t=0}
  \While{$\grad_{\thetab} L_n(\thetab) > \epsilon_{\rm tol}$} 
  {
      \For{$l=1,\dots,L$}
      {
          \For {$i=1,\dots, n$} 
          {
              Sample $\xb_{i,l} \sim Z_l^{-1}  W_l(\xb) d\nu(\xb)$\;
              Propose replica swaps;
          }
          Compute $$\begin{aligned}p_{l,l'} &= \frac{1}{n} \sum_{i=1}^n W_{l'}(\xb_{i,l})\qquad \text{for $l'=1,\ldots, L$, \ \ and}\\
            \Gb_l[\thetab] &= \frac1n\sum_{i=1}^n \nabla_{\thetab} \ell(\xb_{i,l},\thetab)
            \end{aligned}
          $$
      }
      Solve~\eqref{eq:13} for $w_l$, $l=1,\ldots, L$\; 
      Compute $$\grad_{\thetab} L_n(\thetab) = \frac{1}{L} \sum_{l=1}^L \Gb_l(\thetab) w_l$$
      Update $\thetab \leftarrow \thetab - \alpha \grad_{\thetab} L_n(\thetab)$\;
  }
  \KwResult{$\thetab$}
 \caption{Importance Sampled Variational Stochastic Gradient Descent.}
\end{algorithm}
\end{center}

\section{Application: High-dimensional Backward Kolmogorov Equations (BKE)}

Within the framework of statistical mechanics, the evolution of complex physical systems can be described by probability distributions and expectations that solve partial differential equations like the Fokker-Planck equation or the backward Kolmogorov equation (BKE). Because systems of practical interest are often high-dimensional, these PDEs are typically not solved directly---rather we resort to Monte-Carlo sampling methods or molecular dynamics simulations to estimate the system distribution. Our aim here is to investigate whether we can bypass these sampling methods, and go back to solving the relevant PDEs, using tools from ML.

\subsection{ The metastability problem}

We will focus on one specific problem often encountered in practice: how to analyze the dynamics of systems displaying metastability---i.e. evolution that occurs on a wide range of very different time-scales.
Consider in particular a physical system with coordinates
$\Xb_t\in \RR^{d}$ whose evolution is governed by the
Langevin equation
\begin{equation}
 \label{eq:langevin}
 d\Xb_t = -\grad V(\Xb_t) dt + \sqrt{2\beta^{-1}} d\Wb_t.
\end{equation}
Here $V: \RR^d\to[0,\infty)$ is a potential energy function, $\beta>0$, which controls the magnitude of the fluctuations, is typically interpreted as the inverse temperature in physical systems, and $\Wb_t$ is a Wiener process.
This dynamics, or its variant with momentum included, is ubiquitously used to model molecular dynamics in the condensed phase but has also been proposed as a heuristic model for stochastic optimization methods like SGD~\cite{yaida_fluctuation-dissipation_2018-1} and sampling-based optimization schemes~\cite{ma_sampling_2019}.  

In the context of a system with a dynamics governed by~\eqref{eq:langevin}, metastability arises when the system remains confined in some region of its phase space for very long periods of time and seldom makes a transition to another such region.  In general, it is not possible to directly observe the transitions between these metastable states using trajectories generated by~\eqref{eq:langevin} because the state space is very high dimensional in nontrivial cases and these transitions are by definition very infrequent.

Solving a high-dimensional PDE offers an alternative, in principle. Indeed metastability can be characterized mathematically as the property that the spectrum of the infinitesimal generator associated with~\eqref{eq:langevin} has a ``spectral gap'' between a set of low-lying eigenvalues with small magnitude compared to the rest of them---these low lying eigenvalues specify the rates of transition between metastable states, while the associated eigenfunctions describe their mechanism~\cite{bovier_metastability_2002,gaveau_theory_1998}. 
The  eigenvalue/eigenfunction pairs solve the minimization problem
\begin{equation}
    \label{eq:eigen}
    \lambda_k = \min_{\varphi} Z^{-1} \int_{\RR^d} |\nabla \varphi_k(\xb)|^2 e^{-\beta V(\xb)} d\xb, \qquad k\in \NN_0
\end{equation}
where $Z =\int_{\RR^d} e^{-\beta V(\xb)} d\xb $ and successive eigenfunctions are obtained by requiring that they be orthonormal to the previous ones: starting from $\varphi_0 =1$ with $\lambda_0=0$, for $k\in \NN$, we impose
\begin{equation}
    Z^{-1} \int_{\RR^d} \varphi_k(\xb) \varphi_{k'}(\xb) e^{-\beta V(\xb)} d\xb = \delta_{k,k'} \quad \text{for \ $k'=0,\ldots, k$}
\end{equation}
This gives $0=\lambda_0 < \lambda_1 < \cdots $.

While the minimization problem in~\eqref{eq:eigen} fits the framework of~\eqref{eq:ritz1}, in complex systems there may be hundreds or thousands of metastable states, and only a few of them are actually relevant~\cite{cameron_flows_2014}. In this context, it is preferable to focus on one transition of interest at a time. This can be achieved using the potential theoretic framework to metastability \cite{bovier_metastability_2002,bovier_metastability_2006} or transition path theory (TPT) \cite{e_towards_2006,e_transition-path_2010}, and this is the approach we will focus on next. 

\subsection{Potential approach via BKE}

Suppose we want to quantify the average rate and mechanism by which the solution to the Langevin equation~\eqref{eq:langevin} makes a transition from a state $A\subset \RR^{d}$ to a distinct state $B\subset \RR^{d}$.  This can be done by calculating the ``committor function'' $q:\RR^d \to [0,1]$, which gives the probability that a trajectory starting at $\xb$ first reaches $B$ before $A$:
\begin{equation}
 q (\xb) := \mathbb{P}^{\xb}(t_B < t_A)
\end{equation}
where $t_A = \inf \{ t : x(t) \in A\}$ and similarly for $t_B$.
%The committor function $q$ solves the following backward Kolmogorov equation~\cite{e_transition-path_2010} 
% \begin{equation}
%  \label{eq:bke}
%  \begin{cases}
%   L q=0      & \text{ for } \xb \not \in A \cup  B \\
%   q(\xb) = 0 & \text{ for } \xb \in A              \\
%   q(\xb) = 1 & \text{ for } \xb \in B.
%  \end{cases}
% \end{equation}
%We can write a partial differential equation (PDE) for the probability
%$q(\xb,\vb)$ that a trajectory under the dynamics~\eqref{eq:applangevin}
Under the dynamics~\eqref{eq:langevin}, the committor $q(\xb)$ solves the backward Kolmogorov equation 
\begin{equation}
 \label{eq:appbke}
 \begin{cases}
   (L q)(\xb)=0             & \text{ for } \xb \not \in A
   \cup B       \\
  q(\xb) = 0 &\text{ for } \xb \in A \\
  q(\xb) = 1 &\text{ for } \xb \in B.
 \end{cases}
\end{equation}
where  $-L$ is the infinitesimal generator of the process defined
by~\eqref{eq:langevin}:
\begin{equation}
  L q = \grad V \cdot \grad q- \beta^{-1} \lap q.
\end{equation}
It can be shown that, with appropriate choice of $A$ and $B$, $q(\xb)$ can be asymptotically related to a eigenfunction $\varphi_k$ in the low-lying part of the spectrum as $\varphi_k(\xb) =aq(\xb)+b$ for some appropriate choice of $a$ and $b$---we refer the interested reader to \cite{bovier_metastability_2006} for details.  Here we will focus on using the active learning method we propose to solve the backward Kolmogorov equation in~\eqref{eq:appbke} in  high dimension, i.e. in a setup where we would not be able to solve it  using classical numerical PDE methods such as the finite element method. Specifically, our goal in the next sections is  to define a parametric representation of the committor function by a neural network and an objective function that enables us to optimize the parameters in this network via active learning with importance sampling.

\subsection{Variational loss functions for learning the committor}

The committor satisfies a Ritz-type variational principle~\eqref{eq:ritz2} that the can be employed directly as an objective function: That is, the solution to the BKE~\eqref{eq:appbke} is the minimizer of
\begin{equation}
 \label{eq:infqV}
 \inf_{q} C(q) \quad \text{subject to} \ \  q= 0 \ \ \text{in $
   A$} \ \ \text{and} \ \ q  = 1 \ \ \text{in $B$}
\end{equation}
where
\begin{equation}
 \label{eq:objV}
 \begin{aligned}
  C(q) & =  \int_{\RR^d} |\grad q(\xb)|^2 d\nu(\xb)\qquad \text{with} \quad d\nu(\xb) = Z^{-1} e^{-\beta V(\xb)} d\xb
 \end{aligned}
\end{equation}
In the optimization procedure below, it is more tractable to penalize deviations from the boundary conditions rather than impose them as constraints.
Consequently, we add penalty terms in~\eqref{eq:objV} to ensure that the committor has the right values on $A$ and $B$ and the objective function we will use is
\begin{equation}
 \label{eq:objfun}
 \begin{aligned}
  C_\lambda(q) & =  \int_{\RR^d} |\grad q(\xb)|^2 d\nu(\xb)                                                  + \lambda \int_A |q(\xb)|^2 d\nu(\xb)+ \lambda
  \int_B | 1-q(\xb)|^2 d\nu(\xb)                
%               & \equiv \left\<|\grad q|^2 \right\>_\beta +\lambda \left\<
%   |q|^2 1_A\right \>_\beta+ \lambda \left\<
%   |1-q|^2 1_B\right \>_\beta
 \end{aligned}
\end{equation}
where $\lambda>0$ is an adjustable parameter. This objective function is of the type in~\eqref{eq:ritz12} with
\begin{equation}
    \label{eq:lag:com}
    \mathcal{L}(\xb,q) = |\grad q(\xb)|^2                                         + \lambda |q(\xb)|^2 1_A (\xb)+ \lambda
  | 1-q(\xb)|^2 1_B(\xb)
\end{equation}
where $1_{A}(\xb)$ and $1_{B}(\xb)$ denote the indicator functions of sets $A$ and $B$, respectively. 
% where $\<\cdot\>_\beta$ denotes canonical expectation with respect to $d\nu(\xb)=Z^{-1} e^{-\beta V(\xb)} d\xb$, and $1_A$ and $1_B$ are the indicator functions of $A$ and $B$, respectively.

As discussed above, it is natural to model the minimizer of this cost functional as a neural network.
Given some representation $f(\cdot,\thetab)$ with parameter set $\{ \thetab \}_{i=1}^n$, the problem becomes to minimize $C_\lambda$ over this set.
We discuss the specific architectures that we use in applications below, but any neural network architecture is admissible within in this scheme, provided that it gives an output in the range $[0,1]$, which is simple to achieve in practice by passing the output through a sigmoidal function (Appendix~\ref{app:rep}).
Even this condition can be relaxed: We describe an alternative formulation of the committor (cf. ~\cite{lu_exact_2014}) in Appendix~\ref{app:charges} which can be solved with distinct boundary conditions.
The scheme we have described here could be implemented using symmetry functions or collective variables, which we leave for future work.

\subsection{Numerical Experiments}

\subsubsection{M\"uller-Brown potential}

As a proof of concept, we optimize the committor function on the well-studied M\"uller-Brown potential~\cite{muller1979location}.
We consider the dynamics~\eqref{eq:langevin} for a 2D system evolving in a Gaussian mixture potential
\begin{equation}
 \label{eq:v_mb}
 V_{\rm MB}(\xb) = \sum_{k=1}^4 A_k \exp \left( -(\xb-\nu_k)^T \Sigma_k^{-1} (\xb-\nu_k) \right)
\end{equation}
with
\begin{equation}
 \begin{aligned}
   & A = (-200,-100,-170,15)           \\
   & \nu = \begin{pmatrix}1\\0\end{pmatrix},
  \begin{pmatrix}0\\0.5\end{pmatrix},
  \begin{pmatrix}-0.5\\1.5\end{pmatrix},
  \begin{pmatrix}-1\\1\end{pmatrix}           \\
   & \Sigma^{-1} =
  \begin{pmatrix}1 & 0 \\ 0 & 10\end{pmatrix},
  \begin{pmatrix}1 & 0 \\ 0 & 10\end{pmatrix},
  \begin{pmatrix}6.5 & -5.5 \\ -5.5 & 6.5\end{pmatrix},
  \begin{pmatrix}0.7 & 0.3 \\ 0.3 & 0.7\end{pmatrix}
 \end{aligned}
\end{equation}
Our results, shown in Fig.~\ref{fig:mb}, demonstrate the importance sampling is required to converge a robust estimate of the committor. 

While the contours of the committor provide a reasonable guide, a ``committor analysis'' gives a more precise test of convergence. 
To carry this analysis out, we sampled 100 distinct initial configurations from each window (where $q=0.1, 0.2, \dots$) and then ran unbiased Langevin dynamics to compute $\min (t(A), t(B))$.
Histograms of this calculation show that active importance sampling leads to trajectories that reach $B$ before $A$ with the expected probabilities.
However, without importance sampling estimate of $q$ performs poorly.

We represent the committor as a single-hidden layer ReLU network with $m=100$ units.
The output of the ReLU network is passed through a sigmoidal function to compress the range because the committor represents a probability. 
To initialize the representation, we take a discretized linear interpolation between the center of basins $A$ and $B$ and optimize the representation to match the normalized distance along this path. 
At each optimization step, we collect 50 samples from each of the 10 windows in $q$-space. 
We use this sample to estimate the gradient after reweighting, as described above. 
We run the optimization for a total of 1000 optimization steps using 50 samples per window per optimization step with a restraint of $k=100$ in the windowing function. 

To make a systematic comparison, we ran a control experiment in which we used a single unbiased trajectory (that is, no importance sampling) to carry out the optimization.
The total number of samples from this trajectory (10000 optimization steps with 50 samples per step) was chosen to be equal to the total amount of data collected in our importance sampling optimization.
As shown in Fig.~\ref{fig:mb}, this approach does not succeed. 

\subsubsection{Allen-Cahn-type system}

Unlike standard approaches to computing the committor (e.g., finite elements), the algorithm outlined here also succeeds when the input space is high-dimensional.
As a non-trivial test of robustness, we will consider the following example building on a discretized version of the Allen-Cahn equation in two-dimension. 
Let us start from 
\begin{equation}
  \label{eq:pdeCW}
  \partial_t \rho  = D \lap \rho + \rho - \rho^3, \qquad \rho : [0,\infty) \times [0,1]^2 \to \RR
\end{equation}
with the Dirichlet boundary conditions
\begin{equation}
  \label{eq:10}
  \rho(t,z_1,z_2) = +1, \quad \text{for} \ \ z_1=0,1, \qquad \rho(t,z_1,z_2) = -1,
  \quad \text{for} \ \ z_2=0,1,
\end{equation}
The Allen-Cahn equation is the gradient flow in $L^2$ over the energy functional
\begin{equation}
  \label{eq:ECW}
  E(\rho) = \int_{[0,1]^2}\left(
  \tfrac12 D |\grad \rho(\zb)|^2 + \tfrac14 (1-|\rho(\zb)|^2)^2\right) d\zb
\end{equation}
If we take $D$ small enough, this energy admits two minimizers, which are also the stable fixed points of~\eqref{eq:pdeCW} that solve
\begin{equation}
  \label{eq:12}
   D \lap \rho + \rho - \rho^3=0
 \end{equation}
 These fixed points are either mostly $\rho=1$ in the domain, with boundary layer of size $D^{-1/2}$ near $z_2=0,1$, or mostly $\rho=-1$, with boundary layer of size $D^{-1/2}$ near $z_1=0,1$. These two solutions are depicted in Fig.~\ref{fig:cw}.

 \begin{figure}
 \centering
 \parbox[b][4.5cm][c]{0.25\linewidth}{
 \includegraphics[width=\linewidth]{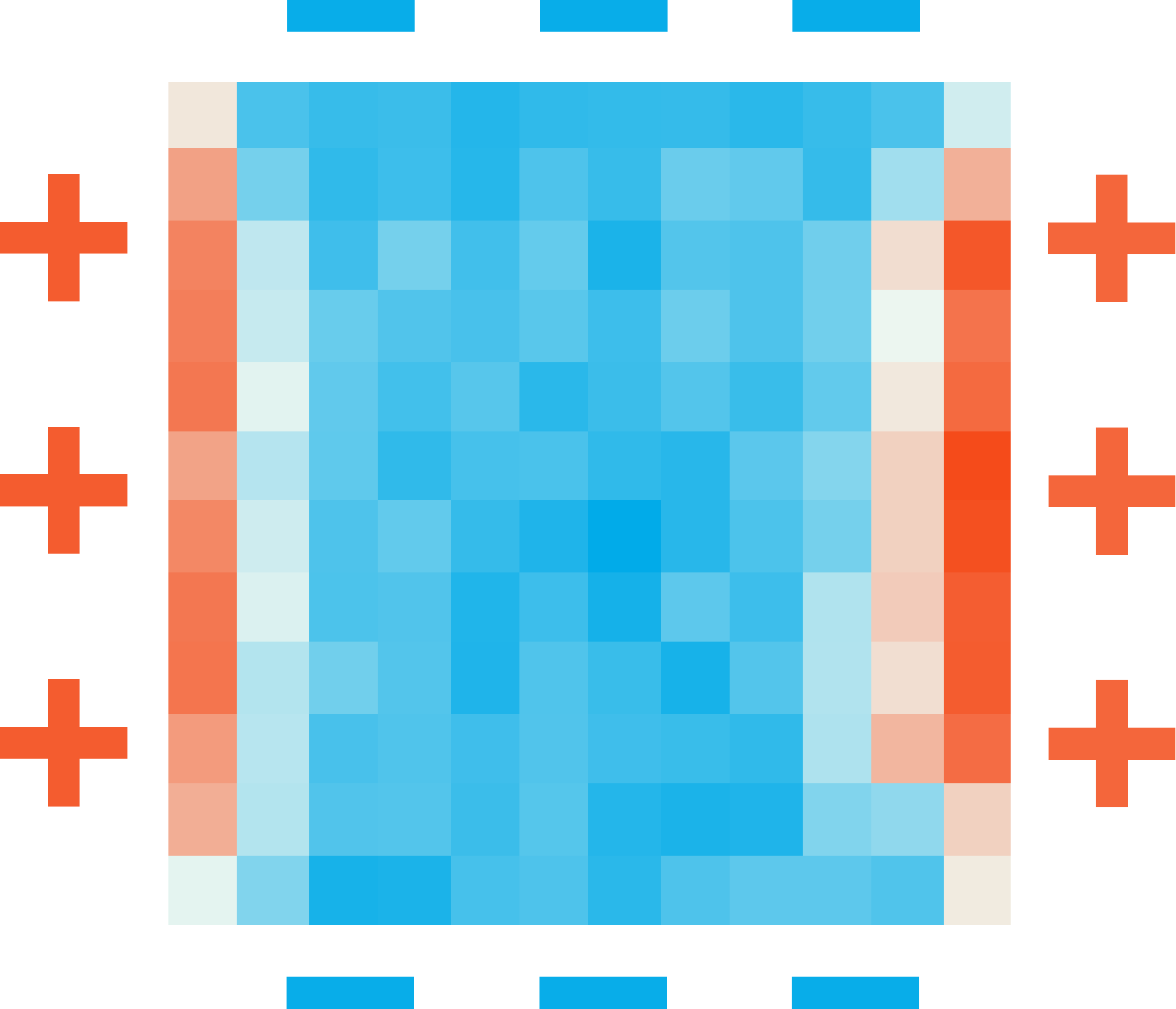}
 }
 \includegraphics[width=0.45\linewidth]{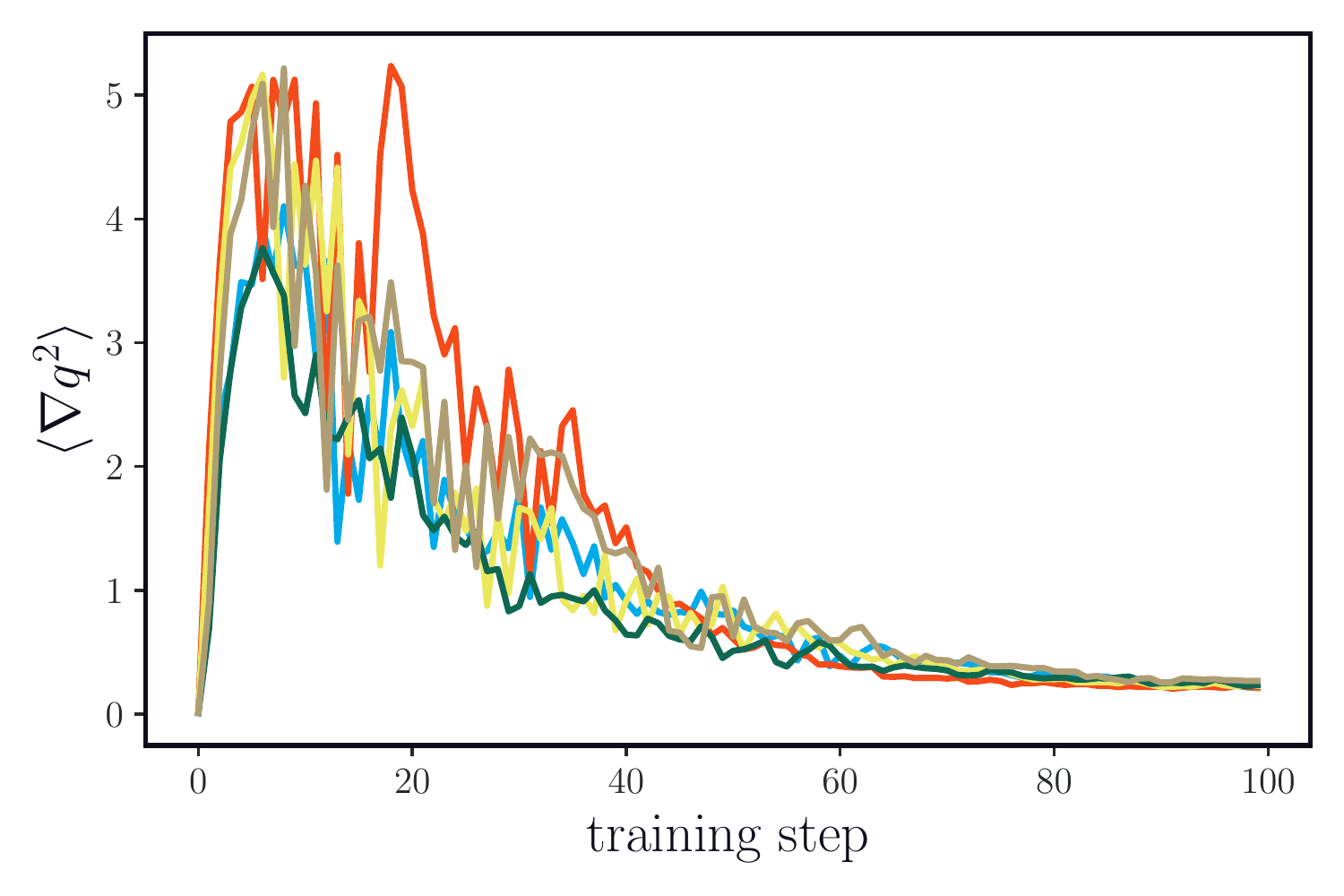}
 \parbox[b][4.5cm][c]{0.25\linewidth}{
 \includegraphics[width=\linewidth]{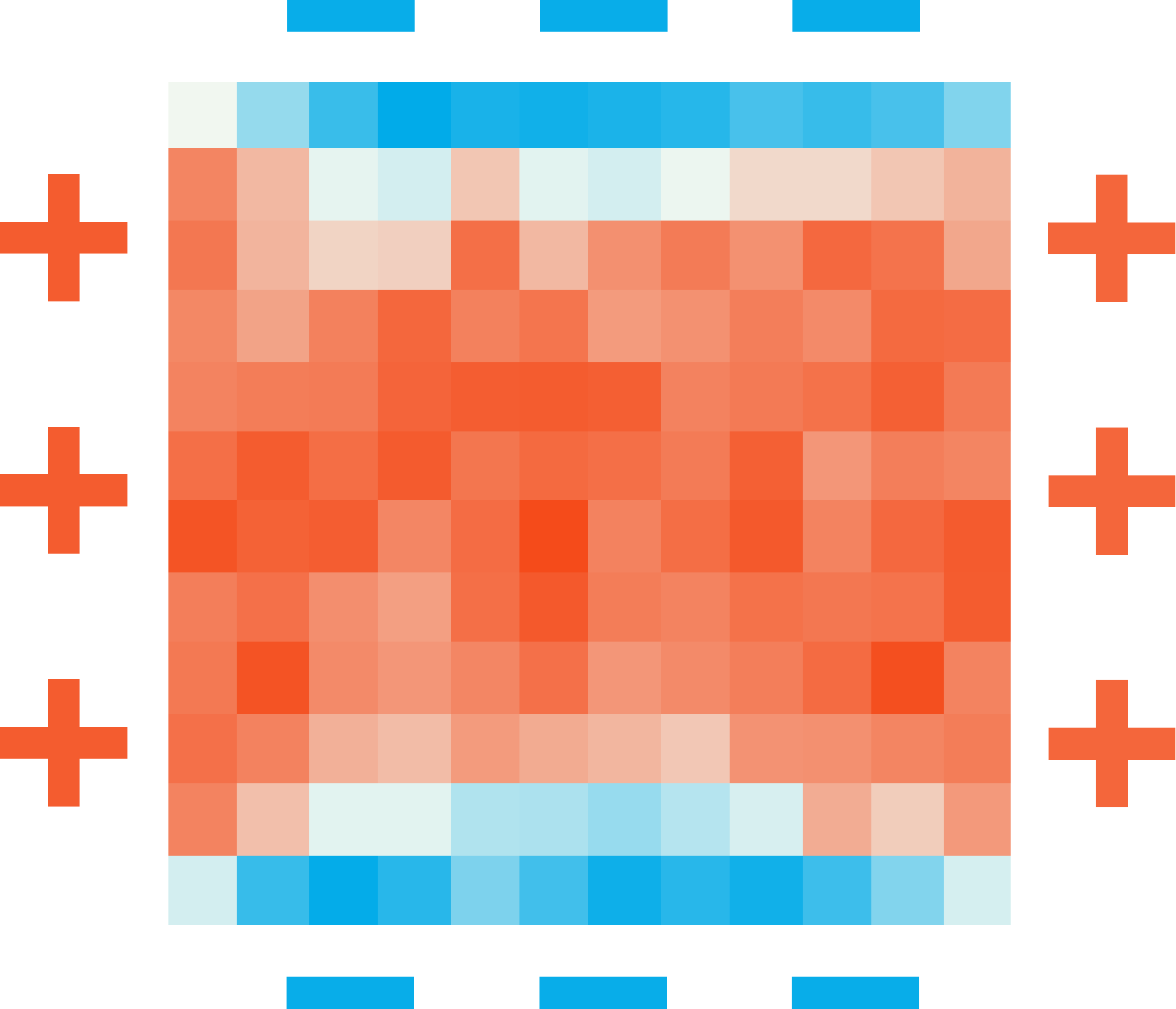}
 }
 \includegraphics[width=1.0\linewidth]{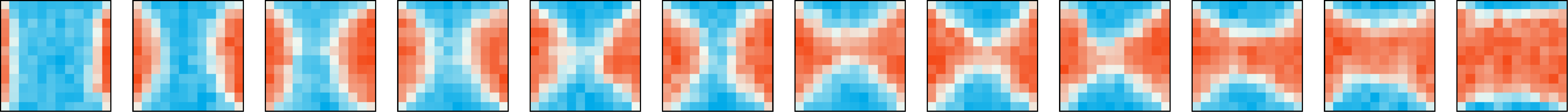}
 \caption{Top left and right: The two metastable solutions of \eqref{eq:cw_sde} with Dirichlet boundary conditions ($\rho=1$ at the left and right boundaries, $\rho=-1$ at the top and bottom boundaries). Top center: Decay of the loss of as a function of training step for 10 runs of the optimization with random initial conditions. Bottom: A sample transition path obtained by sampling the biased ensemble with $(q=0,\dots,1)$. The path shows the characteristic nucleation pathway for a transition between the two metastable states, with the expected hourglass shape.}
 \label{fig:cw}
\end{figure}

To build the model that we will actually study, let us discretize \eqref{eq:pdeCW} on a lattice with spacing $h=1/(N-1)$ and
 introduce
\begin{equation}
  \label{eq:14}
  \rho_{i,j} = \rho(ih,jh), \qquad i,j =1,\ldots, N
\end{equation}
We also add some additive noise to the discretized equation to arrive at the Langevin equation
\begin{equation}
\label{eq:cw_sde}
  d \rho_{i,j} = \left(\rho_{i,j} - \rho_{i,j}^3 + D
    (\Delta_N\rho)_{i,j}\right)dt + \sqrt{2\beta^{-1}}h^{-1} dW_{i,j}
\end{equation}
Here $W_{i,j}$ is set of independent Wiener processes, $\Delta_N$ is the discrete Laplacian,
\begin{equation}
  \label{eq:17}
  (\Delta_N\rho)_{i,j} = h^{-2}
    \left(\rho_{i+1,j} + \rho_{i-1,j} + \rho_{i,j+1}  + \rho_{i,j-1} -
      4\rho_{i,j}\right),
\end{equation}
and the boundary conditions read
\begin{equation}
  \label{eq:16}
  \begin{aligned}
    &\rho_{i,j} = 1, \quad &&\text{for} \ \ i=0,N+1, \quad j=1, \ldots,N,\\
    &\rho_{i,j} = -1, \quad &&\text{for} \ \ j=0,N+1, \quad i=1,
    \ldots,N.
  \end{aligned}
\end{equation}
We also set $\rho_{0,0}=\rho_{0,N+1}=\rho_{N+1,0}=\rho_{N+1,N+1}=0$. 

If we take $D$ and $\beta^{-1}$ small enough, the Langevin equation~\eqref{eq:cw_sde} displays metastability: the solution stays confined for long period of times in regions near the fixed points of the deterministic equation (obtained by setting $\beta^{-1}=0$ in~\eqref{eq:cw_sde}) and only rarely make transition between these regions. This can be confirmed by looking at the equilibrium distribution of~\eqref{eq:cw_sde}:
\begin{equation}
    \label{eq:gibbsAC}
    d\nu(\rhob) = Z^{-1} e^{-\beta V(\rhob)} d\rhob
\end{equation}
where we denote $\rhob = (\rho_{i,j})_{i,j=1}^N$, $Z = \int _{\RR^{2N-2}}e^{-\beta V(\rhob)} d\rhob$  and 
 the potential $V(\rhob)$ is the discrete equivalent to~\eqref{eq:ECW}:
\begin{equation}
  \label{eq:15}
  V(\rhob) =
  h^{-2}\sum_{i,j=1}^N \left(\tfrac12 D |(\grad_N \rho)_{i,j}|^2 +
    \tfrac14 (1-|\rho_{i,j}|^2)^2\right)
\end{equation}
where $\grad_N$
is the discrete gradient so that
\begin{equation}
  \label{eq:17grad}
  |(\grad_N\rho)_{i,j}|^2 = h^{-2}
    \left(\rho_{i+1,j} - \rho_{i,j}\right)^2  + h^{-2}
    \left(\rho_{i,j+1} - \rho_{i,j}\right)^2.
\end{equation}
For small $\beta^{-1}$ the distribution~\eqref{eq:gibbsAC} is nearly atomic on the two minimizers of $V(\rhob)$ shown in Fig.~\ref{fig:cw}., and the question is how do rare transitions occur between these metastable state and at which average rate. This question can be answered by solving the BKE for the committor associated with~\eqref{eq:cw_sde}

We solved this problem using the active learning method outlined before in a situation where $N=12$, i.e. the state space is $12^2=144$ dimensional.
We use a single hidden layer ReLU network with $m=100$ neurons which is passed through a sigmoidal function at the output layer to ensure that the range of $q$ is $(0,1)$.
The network is initialized by linearly interpolating homogeneous configurations in magnetization space, which provides no $\emph{a priori}$ information about the spatial structure of the transition path. 
We carried out the optimization with 12 total windows, including the boundary windows, with a learning rate of $5\times 10^{-2}$ for $5000$ steps with 25 sampling steps per window. 

As shown in Fig.~\ref{fig:cw}, this shows the characteristic pathway for a transition between the two metastable states, with the expected hourglass shape as transition state \cite{kohn2007action} that can also be identified by the string method~\cite{e_string_2002,e_string2007} or the minimum action method in this specific example \cite{weinan2004minimum,heymann2008geometric}.
It should be noted that the initial increase in the loss function arises due to an initial representation of the transition path that is not consistent with the dynamics of the model and that once representative configurations are sampled, the estimate of the loss improves.

\section{Conclusion and Future Work}

The approach we propose here enables optimization in contexts in which the loss function is dominated by data that is exceedingly rare with respect to its equilibrium measure.
While we have both theoretical and numerical evidence that this approach is effective for high-dimensional problems and improves generalization, further evidence from physics applications would bolster our current findings.
In particular, we must test our approach on more complicated systems, like those typically arising in biophysics.
In such systems, there may be multiple pathways connecting two metastable states, a complication that we did not investigate thoroughly here.

In some sense, the promise of machine learning for solving committor equations can be conceptualized by interpreting these problems as classification problems.
In the examples we consider, the primary task directly resembles binary classification in which the network is attempting to find a dividing surface between classes in a high dimensional space.
The isocommittor surface is defined by the dynamical fate of points in this space and collecting data to adequately resolve the location of the boundary is typically impossible without importance sampling. 

While our algorithm and code can easily employ any neural network architecture, we used very simple neural networks for the examples in this paper.
Finding architectures that are well-adapted to a given physical system remains an important challenge~\cite{kearnes_molecular_2016}.
Additionally, there are natural improvements to the implementation of our algorithm: adaptive windowing, more sophisticated reweighting schemes, and exploiting the ``embarrassingly parallel'' structure of the computation to obtain computational speed-ups.

The class of PDEs that we consider here could be generalized to include Ritz-type objectives with forcing terms, as well.
Problems that are driven away from the equilibrium Gibbs distribution pose significant challenges for existing sampling techniques and represent an important target for future work. 

\small

\bibliography{refs}

%apsrev4-2.bst 2019-01-14 (MD) hand-edited version of apsrev4-1.bst
%Control: key (0)
%Control: author (8) initials jnrlst
%Control: editor formatted (1) identically to author
%Control: production of article title (0) allowed
%Control: page (0) single
%Control: year (1) truncated
%Control: production of eprint (0) enabled
\begin{thebibliography}{42}%
\makeatletter
\providecommand \@ifxundefined [1]{%
 \@ifx{#1\undefined}
}%
\providecommand \@ifnum [1]{%
 \ifnum #1\expandafter \@firstoftwo
 \else \expandafter \@secondoftwo
 \fi
}%
\providecommand \@ifx [1]{%
 \ifx #1\expandafter \@firstoftwo
 \else \expandafter \@secondoftwo
 \fi
}%
\providecommand \natexlab [1]{#1}%
\providecommand \enquote  [1]{``#1''}%
\providecommand \bibnamefont  [1]{#1}%
\providecommand \bibfnamefont [1]{#1}%
\providecommand \citenamefont [1]{#1}%
\providecommand \href@noop [0]{\@secondoftwo}%
\providecommand \href [0]{\begingroup \@sanitize@url \@href}%
\providecommand \@href[1]{\@@startlink{#1}\@@href}%
\providecommand \@@href[1]{\endgroup#1\@@endlink}%
\providecommand \@sanitize@url [0]{\catcode `\\12\catcode `\$12\catcode
  `\&12\catcode `\#12\catcode `\^12\catcode `\_12\catcode `\%12\relax}%
\providecommand \@@startlink[1]{}%
\providecommand \@@endlink[0]{}%
\providecommand \url  [0]{\begingroup\@sanitize@url \@url }%
\providecommand \@url [1]{\endgroup\@href {#1}{\urlprefix }}%
\providecommand \urlprefix  [0]{URL }%
\providecommand \Eprint [0]{\href }%
\providecommand \doibase [0]{https://doi.org/}%
\providecommand \selectlanguage [0]{\@gobble}%
\providecommand \bibinfo  [0]{\@secondoftwo}%
\providecommand \bibfield  [0]{\@secondoftwo}%
\providecommand \translation [1]{[#1]}%
\providecommand \BibitemOpen [0]{}%
\providecommand \bibitemStop [0]{}%
\providecommand \bibitemNoStop [0]{.\EOS\space}%
\providecommand \EOS [0]{\spacefactor3000\relax}%
\providecommand \BibitemShut  [1]{\csname bibitem#1\endcsname}%
\let\auto@bib@innerbib\@empty
%</preamble>
\bibitem [{\citenamefont {Carleo}\ \emph {et~al.}(2019)\citenamefont {Carleo},
  \citenamefont {Cirac}, \citenamefont {Cranmer},\ and\ \citenamefont
  {Daudet}}]{carleo_machine_2019}%
  \BibitemOpen
  \bibfield  {author} {\bibinfo {author} {\bibfnamefont {G.}~\bibnamefont
  {Carleo}}, \bibinfo {author} {\bibfnamefont {I.}~\bibnamefont {Cirac}},
  \bibinfo {author} {\bibfnamefont {K.}~\bibnamefont {Cranmer}},\ and\ \bibinfo
  {author} {\bibfnamefont {L.}~\bibnamefont {Daudet}},\ }\bibfield  {title}
  {\bibinfo {title} {Machine learning and the physical sciences*},\ }\href
  {https://doi.org/10/ggd5qv} {\bibfield  {journal} {\bibinfo  {journal} {Rev.
  Mod. Phys.}\ }\textbf {\bibinfo {volume} {91}},\ \bibinfo {pages} {39}
  (\bibinfo {year} {2019})}\BibitemShut {NoStop}%
\bibitem [{\citenamefont {LeCun}\ \emph {et~al.}(2015)\citenamefont {LeCun},
  \citenamefont {Bengio},\ and\ \citenamefont {Hinton}}]{lecun_deep_2015}%
  \BibitemOpen
  \bibfield  {author} {\bibinfo {author} {\bibfnamefont {Y.}~\bibnamefont
  {LeCun}}, \bibinfo {author} {\bibfnamefont {Y.}~\bibnamefont {Bengio}},\ and\
  \bibinfo {author} {\bibfnamefont {G.}~\bibnamefont {Hinton}},\ }\bibfield
  {title} {\bibinfo {title} {Deep learning},\ }\href
  {https://doi.org/10.1038/nature14539} {\bibfield  {journal} {\bibinfo
  {journal} {Nature}\ }\textbf {\bibinfo {volume} {521}},\ \bibinfo {pages}
  {436} (\bibinfo {year} {2015})}\BibitemShut {NoStop}%
\bibitem [{\citenamefont {Donoho}\ and\ \citenamefont
  {Johnstone}(1989)}]{donoho1989}%
  \BibitemOpen
  \bibfield  {author} {\bibinfo {author} {\bibfnamefont {D.~L.}\ \bibnamefont
  {Donoho}}\ and\ \bibinfo {author} {\bibfnamefont {I.~M.}\ \bibnamefont
  {Johnstone}},\ }\bibfield  {title} {\bibinfo {title} {Projection-based
  approximation and a duality with kernel methods},\ }\href
  {https://doi.org/10.1214/aos/1176347004} {\bibfield  {journal} {\bibinfo
  {journal} {Ann. Statist.}\ }\textbf {\bibinfo {volume} {17}},\ \bibinfo
  {pages} {58} (\bibinfo {year} {1989})}\BibitemShut {NoStop}%
\bibitem [{\citenamefont {Toulouse}\ \emph {et~al.}(2016)\citenamefont
  {Toulouse}, \citenamefont {Assaraf},\ and\ \citenamefont
  {Umrigar}}]{hoggan_chapter_2016}%
  \BibitemOpen
  \bibfield  {author} {\bibinfo {author} {\bibfnamefont {J.}~\bibnamefont
  {Toulouse}}, \bibinfo {author} {\bibfnamefont {R.}~\bibnamefont {Assaraf}},\
  and\ \bibinfo {author} {\bibfnamefont {C.~J.}\ \bibnamefont {Umrigar}},\
  }\bibfield  {title} {\bibinfo {title} {Chapter {{Fifteen}} - {{Introduction}}
  to the {{Variational}} and {{Diffusion Monte Carlo Methods}}},\ }in\ \href
  {https://doi.org/10.1016/bs.aiq.2015.07.003} {\emph {\bibinfo {booktitle}
  {Electron {{Correlation}} in {{Molecules}} \textendash{} Ab Initio {{Beyond
  Gaussian Quantum Chemistry}}}}},\ \bibinfo {series} {Advances in {{Quantum
  Chemistry}}}, Vol.~\bibinfo {volume} {73},\ \bibinfo {editor} {edited by\
  \bibinfo {editor} {\bibfnamefont {P.~E.}\ \bibnamefont {Hoggan}}\ and\
  \bibinfo {editor} {\bibfnamefont {T.}~\bibnamefont {Ozdogan}}}\ (\bibinfo
  {publisher} {{Academic Press}},\ \bibinfo {year} {2016})\ pp.\ \bibinfo
  {pages} {285--314}\BibitemShut {NoStop}%
\bibitem [{\citenamefont {Eigel}\ \emph {et~al.}(2019)\citenamefont {Eigel},
  \citenamefont {Schneider}, \citenamefont {Trunschke},\ and\ \citenamefont
  {Wolf}}]{eigel_variational_2019}%
  \BibitemOpen
  \bibfield  {author} {\bibinfo {author} {\bibfnamefont {M.}~\bibnamefont
  {Eigel}}, \bibinfo {author} {\bibfnamefont {R.}~\bibnamefont {Schneider}},
  \bibinfo {author} {\bibfnamefont {P.}~\bibnamefont {Trunschke}},\ and\
  \bibinfo {author} {\bibfnamefont {S.}~\bibnamefont {Wolf}},\ }\bibfield
  {title} {\bibinfo {title} {Variational {{Monte Carlo}}\textemdash bridging
  concepts of machine learning and high-dimensional partial differential
  equations},\ }\href {https://doi.org/10.1007/s10444-019-09723-8} {\bibfield
  {journal} {\bibinfo  {journal} {Advances in Computational Mathematics}\
  }\textbf {\bibinfo {volume} {45}},\ \bibinfo {pages} {2503} (\bibinfo {year}
  {2019})}\BibitemShut {NoStop}%
\bibitem [{\citenamefont {E}\ and\ \citenamefont {Yu}(2017)}]{e_deep_2017}%
  \BibitemOpen
  \bibfield  {author} {\bibinfo {author} {\bibfnamefont {W.}~\bibnamefont {E}}\
  and\ \bibinfo {author} {\bibfnamefont {B.}~\bibnamefont {Yu}},\ }\bibfield
  {title} {\bibinfo {title} {The {{Deep Ritz}} method: {{A}} deep
  learning-based numerical algorithm for solving variational problems},\
  }\href@noop {} {\bibfield  {journal} {\bibinfo  {journal} {arXiv:1710.00211
  [cs, stat]}\ } (\bibinfo {year} {2017})},\ \Eprint
  {https://arxiv.org/abs/1710.00211} {arXiv:1710.00211 [cs, stat]} \BibitemShut
  {NoStop}%
\bibitem [{\citenamefont {Awad}\ \emph {et~al.}(2013)\citenamefont {Awad},
  \citenamefont {Glynn},\ and\ \citenamefont
  {Rubinstein}}]{awad_zero-variance_2013}%
  \BibitemOpen
  \bibfield  {author} {\bibinfo {author} {\bibfnamefont {H.~P.}\ \bibnamefont
  {Awad}}, \bibinfo {author} {\bibfnamefont {P.~W.}\ \bibnamefont {Glynn}},\
  and\ \bibinfo {author} {\bibfnamefont {R.~Y.}\ \bibnamefont {Rubinstein}},\
  }\bibfield  {title} {\bibinfo {title} {Zero-{{Variance Importance Sampling
  Estimators}} for {{Markov Process Expectations}}},\ }\href
  {https://doi.org/10.1287/moor.1120.0569} {\bibfield  {journal} {\bibinfo
  {journal} {Mathematics of Operations Research}\ }\textbf {\bibinfo {volume}
  {38}},\ \bibinfo {pages} {358} (\bibinfo {year} {2013})}\BibitemShut
  {NoStop}%
\bibitem [{\citenamefont {Csiba}\ and\ \citenamefont
  {Richtarik}(2018)}]{csiba_importance_2018}%
  \BibitemOpen
  \bibfield  {author} {\bibinfo {author} {\bibfnamefont {D.}~\bibnamefont
  {Csiba}}\ and\ \bibinfo {author} {\bibfnamefont {P.}~\bibnamefont
  {Richtarik}},\ }\bibfield  {title} {\bibinfo {title} {Importance {{Sampling}}
  for {{Minibatches}}},\ }\href@noop {} {\bibfield  {journal} {\bibinfo
  {journal} {Journal of Machine Learning Research}\ }\textbf {\bibinfo {volume}
  {19}},\ \bibinfo {pages} {21} (\bibinfo {year} {2018})}\BibitemShut {NoStop}%
\bibitem [{\citenamefont {Nesterov}(2012)}]{nesterov_efficiency_2012}%
  \BibitemOpen
  \bibfield  {author} {\bibinfo {author} {\bibfnamefont {Y.}~\bibnamefont
  {Nesterov}},\ }\bibfield  {title} {\bibinfo {title} {Efficiency of
  {{Coordinate Descent Methods}} on {{Huge}}-{{Scale Optimization Problems}}},\
  }\href {https://doi.org/10.1137/100802001} {\bibfield  {journal} {\bibinfo
  {journal} {SIAM Journal on Optimization}\ }\textbf {\bibinfo {volume} {22}},\
  \bibinfo {pages} {341} (\bibinfo {year} {2012})}\BibitemShut {NoStop}%
\bibitem [{\citenamefont {Roux}\ \emph {et~al.}(2012)\citenamefont {Roux},
  \citenamefont {Schmidt},\ and\ \citenamefont {Bach}}]{NIPS2012_4633}%
  \BibitemOpen
  \bibfield  {author} {\bibinfo {author} {\bibfnamefont {N.~L.}\ \bibnamefont
  {Roux}}, \bibinfo {author} {\bibfnamefont {M.}~\bibnamefont {Schmidt}},\ and\
  \bibinfo {author} {\bibfnamefont {F.~R.}\ \bibnamefont {Bach}},\ }\bibfield
  {title} {\bibinfo {title} {A stochastic gradient method with an exponential
  convergence \_{{Rate}} for finite training sets},\ }in\ \href@noop {} {\emph
  {\bibinfo {booktitle} {Advances in Neural Information Processing Systems
  25}}},\ \bibinfo {editor} {edited by\ \bibinfo {editor} {\bibfnamefont
  {F.}~\bibnamefont {Pereira}}, \bibinfo {editor} {\bibfnamefont {C.~J.~C.}\
  \bibnamefont {Burges}}, \bibinfo {editor} {\bibfnamefont {L.}~\bibnamefont
  {Bottou}},\ and\ \bibinfo {editor} {\bibfnamefont {K.~Q.}\ \bibnamefont
  {Weinberger}}}\ (\bibinfo  {publisher} {{Curran Associates, Inc.}},\ \bibinfo
  {year} {2012})\ pp.\ \bibinfo {pages} {2663--2671}\BibitemShut {NoStop}%
\bibitem [{\citenamefont {Johnson}\ and\ \citenamefont
  {Zhang}(2013)}]{NIPS2013_4937}%
  \BibitemOpen
  \bibfield  {author} {\bibinfo {author} {\bibfnamefont {R.}~\bibnamefont
  {Johnson}}\ and\ \bibinfo {author} {\bibfnamefont {T.}~\bibnamefont
  {Zhang}},\ }\bibfield  {title} {\bibinfo {title} {Accelerating stochastic
  gradient descent using predictive variance reduction},\ }in\ \href@noop {}
  {\emph {\bibinfo {booktitle} {Advances in Neural Information Processing
  Systems 26}}},\ \bibinfo {editor} {edited by\ \bibinfo {editor}
  {\bibfnamefont {C.~J.~C.}\ \bibnamefont {Burges}}, \bibinfo {editor}
  {\bibfnamefont {L.}~\bibnamefont {Bottou}}, \bibinfo {editor} {\bibfnamefont
  {M.}~\bibnamefont {Welling}}, \bibinfo {editor} {\bibfnamefont
  {Z.}~\bibnamefont {Ghahramani}},\ and\ \bibinfo {editor} {\bibfnamefont
  {K.~Q.}\ \bibnamefont {Weinberger}}}\ (\bibinfo  {publisher} {{Curran
  Associates, Inc.}},\ \bibinfo {year} {2013})\ pp.\ \bibinfo {pages}
  {315--323}\BibitemShut {NoStop}%
\bibitem [{\citenamefont {Fan}\ \emph {et~al.}(2010)\citenamefont {Fan},
  \citenamefont {Xu},\ and\ \citenamefont {Shelton}}]{fan_importance_2010}%
  \BibitemOpen
  \bibfield  {author} {\bibinfo {author} {\bibfnamefont {Y.}~\bibnamefont
  {Fan}}, \bibinfo {author} {\bibfnamefont {J.}~\bibnamefont {Xu}},\ and\
  \bibinfo {author} {\bibfnamefont {C.~R.}\ \bibnamefont {Shelton}},\
  }\bibfield  {title} {\bibinfo {title} {Importance {{Sampling}} for
  {{Continuous Time Bayesian Networks}}},\ }\href@noop {} {\bibfield  {journal}
  {\bibinfo  {journal} {Journal of Machine Learning Research}\ }\textbf
  {\bibinfo {volume} {11}},\ \bibinfo {pages} {2115} (\bibinfo {year}
  {2010})}\BibitemShut {NoStop}%
\bibitem [{\citenamefont {Han}\ \emph {et~al.}(2019)\citenamefont {Han},
  \citenamefont {Zhang},\ and\ \citenamefont {E}}]{han_solving_2019}%
  \BibitemOpen
  \bibfield  {author} {\bibinfo {author} {\bibfnamefont {J.}~\bibnamefont
  {Han}}, \bibinfo {author} {\bibfnamefont {L.}~\bibnamefont {Zhang}},\ and\
  \bibinfo {author} {\bibfnamefont {W.}~\bibnamefont {E}},\ }\bibfield  {title}
  {\bibinfo {title} {Solving many-electron {{Schr\"odinger}} equation using
  deep neural networks},\ }\href {https://doi.org/10.1016/j.jcp.2019.108929}
  {\bibfield  {journal} {\bibinfo  {journal} {Journal of Computational
  Physics}\ }\textbf {\bibinfo {volume} {399}},\ \bibinfo {pages} {108929}
  (\bibinfo {year} {2019})}\BibitemShut {NoStop}%
\bibitem [{\citenamefont {Hermann}\ \emph {et~al.}(2020)\citenamefont
  {Hermann}, \citenamefont {Sch{\"a}tzle},\ and\ \citenamefont
  {No{\'e}}}]{hermann_deep-neural-network_2020}%
  \BibitemOpen
  \bibfield  {author} {\bibinfo {author} {\bibfnamefont {J.}~\bibnamefont
  {Hermann}}, \bibinfo {author} {\bibfnamefont {Z.}~\bibnamefont
  {Sch{\"a}tzle}},\ and\ \bibinfo {author} {\bibfnamefont {F.}~\bibnamefont
  {No{\'e}}},\ }\bibfield  {title} {\bibinfo {title} {Deep-neural-network
  solution of the electronic {{Schr\"odinger}} equation},\ }\href
  {https://doi.org/10.1038/s41557-020-0544-y} {\bibfield  {journal} {\bibinfo
  {journal} {Nature Chemistry}\ }\textbf {\bibinfo {volume} {12}},\ \bibinfo
  {pages} {891} (\bibinfo {year} {2020})}\BibitemShut {NoStop}%
\bibitem [{\citenamefont {Pfau}\ \emph {et~al.}(2020)\citenamefont {Pfau},
  \citenamefont {Spencer}, \citenamefont {Matthews},\ and\ \citenamefont
  {Foulkes}}]{pfau_ab_2020}%
  \BibitemOpen
  \bibfield  {author} {\bibinfo {author} {\bibfnamefont {D.}~\bibnamefont
  {Pfau}}, \bibinfo {author} {\bibfnamefont {J.~S.}\ \bibnamefont {Spencer}},
  \bibinfo {author} {\bibfnamefont {A.~G. D.~G.}\ \bibnamefont {Matthews}},\
  and\ \bibinfo {author} {\bibfnamefont {W.~M.~C.}\ \bibnamefont {Foulkes}},\
  }\bibfield  {title} {\bibinfo {title} {Ab initio solution of the
  many-electron {{Schr}}\textbackslash "odinger equation with deep neural
  networks},\ }\href {https://doi.org/10.1103/PhysRevResearch.2.033429}
  {\bibfield  {journal} {\bibinfo  {journal} {Physical Review Research}\
  }\textbf {\bibinfo {volume} {2}},\ \bibinfo {pages} {033429} (\bibinfo {year}
  {2020})}\BibitemShut {NoStop}%
\bibitem [{\citenamefont {Bolhuis}\ \emph {et~al.}(2002)\citenamefont
  {Bolhuis}, \citenamefont {Chandler}, \citenamefont {Dellago},\ and\
  \citenamefont {Geissler}}]{bolhuis_transition_2002}%
  \BibitemOpen
  \bibfield  {author} {\bibinfo {author} {\bibfnamefont {P.~G.}\ \bibnamefont
  {Bolhuis}}, \bibinfo {author} {\bibfnamefont {D.}~\bibnamefont {Chandler}},
  \bibinfo {author} {\bibfnamefont {C.}~\bibnamefont {Dellago}},\ and\ \bibinfo
  {author} {\bibfnamefont {P.~L.}\ \bibnamefont {Geissler}},\ }\bibfield
  {title} {\bibinfo {title} {Transition path sampling: Throwing ropes over
  rough mountain passes, in the dark.},\ }\href
  {https://doi.org/10.1146/annurev.physchem.53.082301.113146} {\bibfield
  {journal} {\bibinfo  {journal} {Annu. Rev. Phys. Chem.}\ }\textbf {\bibinfo
  {volume} {53}},\ \bibinfo {pages} {291} (\bibinfo {year} {2002})}\BibitemShut
  {NoStop}%
\bibitem [{\citenamefont {Maragliano}\ \emph {et~al.}(2006)\citenamefont
  {Maragliano}, \citenamefont {Fischer}, \citenamefont {{Vanden-Eijnden}},\
  and\ \citenamefont {Ciccotti}}]{maragliano_string_2006}%
  \BibitemOpen
  \bibfield  {author} {\bibinfo {author} {\bibfnamefont {L.}~\bibnamefont
  {Maragliano}}, \bibinfo {author} {\bibfnamefont {A.}~\bibnamefont {Fischer}},
  \bibinfo {author} {\bibfnamefont {E.}~\bibnamefont {{Vanden-Eijnden}}},\ and\
  \bibinfo {author} {\bibfnamefont {G.}~\bibnamefont {Ciccotti}},\ }\bibfield
  {title} {\bibinfo {title} {String method in collective variables: {{Minimum}}
  free energy paths and isocommittor surfaces},\ }\href
  {https://doi.org/10.1063/1.2212942} {\bibfield  {journal} {\bibinfo
  {journal} {J. Chem. Phys.}\ }\textbf {\bibinfo {volume} {125}},\ \bibinfo
  {pages} {024106} (\bibinfo {year} {2006})}\BibitemShut {NoStop}%
\bibitem [{\citenamefont {Bovier}\ \emph {et~al.}(2002)\citenamefont {Bovier},
  \citenamefont {Eckhoff}, \citenamefont {Gayrard},\ and\ \citenamefont
  {Klein}}]{bovier_metastability_2002}%
  \BibitemOpen
  \bibfield  {author} {\bibinfo {author} {\bibfnamefont {A.}~\bibnamefont
  {Bovier}}, \bibinfo {author} {\bibfnamefont {M.}~\bibnamefont {Eckhoff}},
  \bibinfo {author} {\bibfnamefont {V.}~\bibnamefont {Gayrard}},\ and\ \bibinfo
  {author} {\bibfnamefont {M.}~\bibnamefont {Klein}},\ }\bibfield  {title}
  {\bibinfo {title} {Metastability and {{Low Lying Spectra}} in {{Reversible
  Markov Chains}}},\ }\href {https://doi.org/10.1007/s002200200609} {\bibfield
  {journal} {\bibinfo  {journal} {Communications in Mathematical Physics}\
  }\textbf {\bibinfo {volume} {228}},\ \bibinfo {pages} {219} (\bibinfo {year}
  {2002})}\BibitemShut {NoStop}%
\bibitem [{\citenamefont {E}\ and\ \citenamefont
  {{Vanden-Eijnden}}(2006)}]{e_towards_2006}%
  \BibitemOpen
  \bibfield  {author} {\bibinfo {author} {\bibfnamefont {W.}~\bibnamefont {E}}\
  and\ \bibinfo {author} {\bibfnamefont {E.}~\bibnamefont {{Vanden-Eijnden}}},\
  }\bibfield  {title} {\bibinfo {title} {Towards a {{Theory}} of {{Transition
  Paths}}},\ }\href {https://doi.org/10/b3pwgn} {\bibfield  {journal} {\bibinfo
   {journal} {Journal of Statistical Physics}\ }\textbf {\bibinfo {volume}
  {123}},\ \bibinfo {pages} {503} (\bibinfo {year} {2006})}\BibitemShut
  {NoStop}%
\bibitem [{\citenamefont {E}\ and\ \citenamefont
  {{Vanden-Eijnden}}(2010)}]{e_transition-path_2010}%
  \BibitemOpen
  \bibfield  {author} {\bibinfo {author} {\bibfnamefont {W.}~\bibnamefont {E}}\
  and\ \bibinfo {author} {\bibfnamefont {E.}~\bibnamefont {{Vanden-Eijnden}}},\
  }\bibfield  {title} {\bibinfo {title} {Transition-{{Path Theory}} and
  {{Path}}-{{Finding Algorithms}} for the {{Study}} of {{Rare Events}}},\
  }\href {https://doi.org/10.1146/annurev.physchem.040808.090412} {\bibfield
  {journal} {\bibinfo  {journal} {Annual Review of Physical Chemistry}\
  }\textbf {\bibinfo {volume} {61}},\ \bibinfo {pages} {391} (\bibinfo {year}
  {2010})}\BibitemShut {NoStop}%
\bibitem [{\citenamefont {Khoo}\ \emph {et~al.}(2018)\citenamefont {Khoo},
  \citenamefont {Lu},\ and\ \citenamefont {Ying}}]{khoo_solving_2018}%
  \BibitemOpen
  \bibfield  {author} {\bibinfo {author} {\bibfnamefont {Y.}~\bibnamefont
  {Khoo}}, \bibinfo {author} {\bibfnamefont {J.}~\bibnamefont {Lu}},\ and\
  \bibinfo {author} {\bibfnamefont {L.}~\bibnamefont {Ying}},\ }\bibfield
  {title} {\bibinfo {title} {Solving for high dimensional committor functions
  using artificial neural networks},\ }\href@noop {} {\bibfield  {journal}
  {\bibinfo  {journal} {arXiv}\ } (\bibinfo {year} {2018})},\ \Eprint
  {https://arxiv.org/abs/1802.10275v1} {arXiv:1802.10275v1} \BibitemShut
  {NoStop}%
\bibitem [{\citenamefont {Li}\ \emph {et~al.}(2019)\citenamefont {Li},
  \citenamefont {Lin},\ and\ \citenamefont {Ren}}]{li_computing_2019}%
  \BibitemOpen
  \bibfield  {author} {\bibinfo {author} {\bibfnamefont {Q.}~\bibnamefont
  {Li}}, \bibinfo {author} {\bibfnamefont {B.}~\bibnamefont {Lin}},\ and\
  \bibinfo {author} {\bibfnamefont {W.}~\bibnamefont {Ren}},\ }\bibfield
  {title} {\bibinfo {title} {Computing {{Committor Functions}} for the
  {{Study}} of {{Rare Events Using Deep Learning}}},\ }\href
  {https://doi.org/10.1063/1.5110439} {\bibfield  {journal} {\bibinfo
  {journal} {The Journal of Chemical Physics}\ }\textbf {\bibinfo {volume}
  {151}},\ \bibinfo {pages} {054112} (\bibinfo {year} {2019})},\ \Eprint
  {https://arxiv.org/abs/1906.06285} {arXiv:1906.06285} \BibitemShut {NoStop}%
\bibitem [{\citenamefont {Mandt}\ \emph {et~al.}(2017)\citenamefont {Mandt},
  \citenamefont {Hoffman},\ and\ \citenamefont {Blei}}]{mandt_stochastic_2017}%
  \BibitemOpen
  \bibfield  {author} {\bibinfo {author} {\bibfnamefont {S.}~\bibnamefont
  {Mandt}}, \bibinfo {author} {\bibfnamefont {M.~D.}\ \bibnamefont {Hoffman}},\
  and\ \bibinfo {author} {\bibfnamefont {D.~M.}\ \bibnamefont {Blei}},\
  }\bibfield  {title} {\bibinfo {title} {Stochastic gradient descent as
  approximate bayesian inference},\ }\href@noop {} {\bibfield  {journal}
  {\bibinfo  {journal} {Journal of Machine Learning Research}\ }\textbf
  {\bibinfo {volume} {18}},\ \bibinfo {pages} {1} (\bibinfo {year}
  {2017})}\BibitemShut {NoStop}%
\bibitem [{\citenamefont {Robbins}\ and\ \citenamefont
  {Monro}(1951)}]{robbins1951}%
  \BibitemOpen
  \bibfield  {author} {\bibinfo {author} {\bibfnamefont {H.}~\bibnamefont
  {Robbins}}\ and\ \bibinfo {author} {\bibfnamefont {S.}~\bibnamefont
  {Monro}},\ }\bibfield  {title} {\bibinfo {title} {A stochastic approximation
  method},\ }\href {https://doi.org/10.1214/aoms/1177729586} {\bibfield
  {journal} {\bibinfo  {journal} {Ann. Math. Statist.}\ }\textbf {\bibinfo
  {volume} {22}},\ \bibinfo {pages} {400} (\bibinfo {year} {1951})}\BibitemShut
  {NoStop}%
\bibitem [{\citenamefont {Torrie}\ and\ \citenamefont
  {Valleau}(1977)}]{torrie1977}%
  \BibitemOpen
  \bibfield  {author} {\bibinfo {author} {\bibfnamefont {G.}~\bibnamefont
  {Torrie}}\ and\ \bibinfo {author} {\bibfnamefont {J.}~\bibnamefont
  {Valleau}},\ }\bibfield  {title} {\bibinfo {title} {Nonphysical sampling
  distributions in monte carlo free-energy estimation: Umbrella sampling},\
  }\href {https://doi.org/https://doi.org/10.1016/0021-9991(77)90121-8}
  {\bibfield  {journal} {\bibinfo  {journal} {Journal of Computational
  Physics}\ }\textbf {\bibinfo {volume} {23}},\ \bibinfo {pages} {187 }
  (\bibinfo {year} {1977})}\BibitemShut {NoStop}%
\bibitem [{\citenamefont {Dinner}\ \emph {et~al.}(2020)\citenamefont {Dinner},
  \citenamefont {Thiede}, \citenamefont {Koten},\ and\ \citenamefont
  {Weare}}]{dinner2020}%
  \BibitemOpen
  \bibfield  {author} {\bibinfo {author} {\bibfnamefont {A.~R.}\ \bibnamefont
  {Dinner}}, \bibinfo {author} {\bibfnamefont {E.~H.}\ \bibnamefont {Thiede}},
  \bibinfo {author} {\bibfnamefont {B.~V.}\ \bibnamefont {Koten}},\ and\
  \bibinfo {author} {\bibfnamefont {J.}~\bibnamefont {Weare}},\ }\bibfield
  {title} {\bibinfo {title} {Stratification as a general variance reduction
  method for markov chain monte carlo},\ }\href
  {https://doi.org/10.1137/18M122964X} {\bibfield  {journal} {\bibinfo
  {journal} {SIAM/ASA Journal on Uncertainty Quantification}\ }\textbf
  {\bibinfo {volume} {8}},\ \bibinfo {pages} {1139} (\bibinfo {year} {2020})},\
  \Eprint {https://arxiv.org/abs/https://doi.org/10.1137/18M122964X}
  {https://doi.org/10.1137/18M122964X} \BibitemShut {NoStop}%
\bibitem [{\citenamefont {Swendsen}\ and\ \citenamefont
  {Wang}(1986)}]{swendsen1986}%
  \BibitemOpen
  \bibfield  {author} {\bibinfo {author} {\bibfnamefont {R.~H.}\ \bibnamefont
  {Swendsen}}\ and\ \bibinfo {author} {\bibfnamefont {J.-S.}\ \bibnamefont
  {Wang}},\ }\bibfield  {title} {\bibinfo {title} {Replica monte carlo
  simulation of spin-glasses},\ }\href
  {https://doi.org/10.1103/PhysRevLett.57.2607} {\bibfield  {journal} {\bibinfo
   {journal} {Phys. Rev. Lett.}\ }\textbf {\bibinfo {volume} {57}},\ \bibinfo
  {pages} {2607} (\bibinfo {year} {1986})}\BibitemShut {NoStop}%
\bibitem [{\citenamefont {Fukunishi}\ \emph {et~al.}(2002)\citenamefont
  {Fukunishi}, \citenamefont {Watanabe},\ and\ \citenamefont
  {Takada}}]{fukunishi_hamiltonian_2002}%
  \BibitemOpen
  \bibfield  {author} {\bibinfo {author} {\bibfnamefont {H.}~\bibnamefont
  {Fukunishi}}, \bibinfo {author} {\bibfnamefont {O.}~\bibnamefont
  {Watanabe}},\ and\ \bibinfo {author} {\bibfnamefont {S.}~\bibnamefont
  {Takada}},\ }\bibfield  {title} {\bibinfo {title} {On the {{Hamiltonian}}
  replica exchange method for efficient sampling of biomolecular systems:
  {{Application}} to protein structure prediction},\ }\href
  {https://doi.org/10.1063/1.1472510} {\bibfield  {journal} {\bibinfo
  {journal} {The Journal of Chemical Physics}\ }\textbf {\bibinfo {volume}
  {116}},\ \bibinfo {pages} {9058} (\bibinfo {year} {2002})}\BibitemShut
  {NoStop}%
\bibitem [{\citenamefont {Thiede}\ \emph {et~al.}(2016)\citenamefont {Thiede},
  \citenamefont {Van~Koten}, \citenamefont {Weare},\ and\ \citenamefont
  {Dinner}}]{thiede2016eigenvector}%
  \BibitemOpen
  \bibfield  {author} {\bibinfo {author} {\bibfnamefont {E.~H.}\ \bibnamefont
  {Thiede}}, \bibinfo {author} {\bibfnamefont {B.}~\bibnamefont {Van~Koten}},
  \bibinfo {author} {\bibfnamefont {J.}~\bibnamefont {Weare}},\ and\ \bibinfo
  {author} {\bibfnamefont {A.~R.}\ \bibnamefont {Dinner}},\ }\bibfield  {title}
  {\bibinfo {title} {Eigenvector method for umbrella sampling enables error
  analysis},\ }\href@noop {} {\bibfield  {journal} {\bibinfo  {journal} {The
  Journal of chemical physics}\ }\textbf {\bibinfo {volume} {145}},\ \bibinfo
  {pages} {084115} (\bibinfo {year} {2016})}\BibitemShut {NoStop}%
\bibitem [{\citenamefont {Yaida}(2018)}]{yaida_fluctuation-dissipation_2018-1}%
  \BibitemOpen
  \bibfield  {author} {\bibinfo {author} {\bibfnamefont {S.}~\bibnamefont
  {Yaida}},\ }\bibfield  {title} {\bibinfo {title} {Fluctuation-dissipation
  relations for stochastic gradient descent},\ }\href@noop {} {\bibfield
  {journal} {\bibinfo  {journal} {arXiv:1810.00004 [cs, stat]}\ } (\bibinfo
  {year} {2018})},\ \Eprint {https://arxiv.org/abs/1810.00004}
  {arXiv:1810.00004 [cs, stat]} \BibitemShut {NoStop}%
\bibitem [{\citenamefont {Ma}\ \emph {et~al.}(2019)\citenamefont {Ma},
  \citenamefont {Chen}, \citenamefont {Jin}, \citenamefont {Flammarion},\ and\
  \citenamefont {Jordan}}]{ma_sampling_2019}%
  \BibitemOpen
  \bibfield  {author} {\bibinfo {author} {\bibfnamefont {Y.-A.}\ \bibnamefont
  {Ma}}, \bibinfo {author} {\bibfnamefont {Y.}~\bibnamefont {Chen}}, \bibinfo
  {author} {\bibfnamefont {C.}~\bibnamefont {Jin}}, \bibinfo {author}
  {\bibfnamefont {N.}~\bibnamefont {Flammarion}},\ and\ \bibinfo {author}
  {\bibfnamefont {M.~I.}\ \bibnamefont {Jordan}},\ }\bibfield  {title}
  {\bibinfo {title} {Sampling can be faster than optimization},\ }\href
  {https://doi.org/10.1073/pnas.1820003116} {\bibfield  {journal} {\bibinfo
  {journal} {Proceedings of the National Academy of Sciences}\ }\textbf
  {\bibinfo {volume} {116}},\ \bibinfo {pages} {20881} (\bibinfo {year}
  {2019})}\BibitemShut {NoStop}%
\bibitem [{\citenamefont {Gaveau}\ and\ \citenamefont
  {Schulman}(1998)}]{gaveau_theory_1998}%
  \BibitemOpen
  \bibfield  {author} {\bibinfo {author} {\bibfnamefont {B.}~\bibnamefont
  {Gaveau}}\ and\ \bibinfo {author} {\bibfnamefont {L.~S.}\ \bibnamefont
  {Schulman}},\ }\bibfield  {title} {\bibinfo {title} {Theory of nonequilibrium
  first-order phase transitions for stochastic dynamics},\ }\href
  {https://doi.org/10.1063/1.532394} {\bibfield  {journal} {\bibinfo  {journal}
  {Journal of Mathematical Physics}\ }\textbf {\bibinfo {volume} {39}},\
  \bibinfo {pages} {1517} (\bibinfo {year} {1998})}\BibitemShut {NoStop}%
\bibitem [{\citenamefont {Cameron}\ and\ \citenamefont
  {{Vanden-Eijnden}}(2014)}]{cameron_flows_2014}%
  \BibitemOpen
  \bibfield  {author} {\bibinfo {author} {\bibfnamefont {M.}~\bibnamefont
  {Cameron}}\ and\ \bibinfo {author} {\bibfnamefont {E.}~\bibnamefont
  {{Vanden-Eijnden}}},\ }\bibfield  {title} {\bibinfo {title} {Flows in
  {{Complex Networks}}: {{Theory}}, {{Algorithms}}, and {{Application}} to
  {{Lennard}}\textendash{{Jones Cluster Rearrangement}}},\ }\href
  {https://doi.org/10.1007/s10955-014-0997-8} {\bibfield  {journal} {\bibinfo
  {journal} {Journal of Statistical Physics}\ }\textbf {\bibinfo {volume}
  {156}},\ \bibinfo {pages} {427} (\bibinfo {year} {2014})}\BibitemShut
  {NoStop}%
\bibitem [{\citenamefont {Bovier}(2006)}]{bovier_metastability_2006}%
  \BibitemOpen
  \bibfield  {author} {\bibinfo {author} {\bibfnamefont {A.}~\bibnamefont
  {Bovier}},\ }\bibfield  {title} {\bibinfo {title} {Metastability: A potential
  theoretic approach},\ }\href@noop {} {\bibfield  {journal} {\bibinfo
  {journal} {Proceedings of the International Congress of Mathematicians}\ ,\
  \bibinfo {pages} {20}} (\bibinfo {year} {2006})}\BibitemShut {NoStop}%
\bibitem [{\citenamefont {Lu}\ and\ \citenamefont
  {{Vanden-Eijnden}}(2014)}]{lu_exact_2014}%
  \BibitemOpen
  \bibfield  {author} {\bibinfo {author} {\bibfnamefont {J.}~\bibnamefont
  {Lu}}\ and\ \bibinfo {author} {\bibfnamefont {E.}~\bibnamefont
  {{Vanden-Eijnden}}},\ }\bibfield  {title} {\bibinfo {title} {Exact dynamical
  coarse-graining without time-scale separation},\ }\href
  {https://doi.org/10/gf3xf4} {\bibfield  {journal} {\bibinfo  {journal} {The
  Journal of Chemical Physics}\ }\textbf {\bibinfo {volume} {141}},\ \bibinfo
  {pages} {044109} (\bibinfo {year} {2014})}\BibitemShut {NoStop}%
\bibitem [{\citenamefont {M{\"u}ller}\ and\ \citenamefont
  {Brown}(1979)}]{muller1979location}%
  \BibitemOpen
  \bibfield  {author} {\bibinfo {author} {\bibfnamefont {K.}~\bibnamefont
  {M{\"u}ller}}\ and\ \bibinfo {author} {\bibfnamefont {L.~D.}\ \bibnamefont
  {Brown}},\ }\bibfield  {title} {\bibinfo {title} {Location of saddle points
  and minimum energy paths by a constrained simplex optimization procedure},\
  }\href@noop {} {\bibfield  {journal} {\bibinfo  {journal} {Theoretica chimica
  acta}\ }\textbf {\bibinfo {volume} {53}},\ \bibinfo {pages} {75} (\bibinfo
  {year} {1979})}\BibitemShut {NoStop}%
\bibitem [{\citenamefont {Kohn}\ \emph {et~al.}(2007)\citenamefont {Kohn},
  \citenamefont {Otto}, \citenamefont {Reznikoff},\ and\ \citenamefont
  {{Vanden-Eijnden}}}]{kohn2007action}%
  \BibitemOpen
  \bibfield  {author} {\bibinfo {author} {\bibfnamefont {R.~V.}\ \bibnamefont
  {Kohn}}, \bibinfo {author} {\bibfnamefont {F.}~\bibnamefont {Otto}}, \bibinfo
  {author} {\bibfnamefont {M.~G.}\ \bibnamefont {Reznikoff}},\ and\ \bibinfo
  {author} {\bibfnamefont {E.}~\bibnamefont {{Vanden-Eijnden}}},\ }\bibfield
  {title} {\bibinfo {title} {Action minimization and sharp-interface limits for
  the stochastic {{Allen}}-{{Cahn}} equation},\ }\href@noop {} {\bibfield
  {journal} {\bibinfo  {journal} {Communications on Pure and Applied
  Mathematics: A Journal Issued by the Courant Institute of Mathematical
  Sciences}\ }\textbf {\bibinfo {volume} {60}},\ \bibinfo {pages} {393}
  (\bibinfo {year} {2007})}\BibitemShut {NoStop}%
\bibitem [{\citenamefont {E}\ \emph {et~al.}(2002)\citenamefont {E},
  \citenamefont {Ren},\ and\ \citenamefont {{Vanden-Eijnden}}}]{e_string_2002}%
  \BibitemOpen
  \bibfield  {author} {\bibinfo {author} {\bibfnamefont {W.}~\bibnamefont {E}},
  \bibinfo {author} {\bibfnamefont {W.}~\bibnamefont {Ren}},\ and\ \bibinfo
  {author} {\bibfnamefont {E.}~\bibnamefont {{Vanden-Eijnden}}},\ }\bibfield
  {title} {\bibinfo {title} {String method for the study of rare events},\
  }\href {https://doi.org/10.1103/PhysRevB.66.052301} {\bibfield  {journal}
  {\bibinfo  {journal} {Phys. Rev. B}\ }\textbf {\bibinfo {volume} {66}},\
  \bibinfo {pages} {052301} (\bibinfo {year} {2002})}\BibitemShut {NoStop}%
\bibitem [{\citenamefont {E}\ \emph {et~al.}(2007)\citenamefont {E},
  \citenamefont {Ren},\ and\ \citenamefont {Vanden-Eijnden}}]{e_string2007}%
  \BibitemOpen
  \bibfield  {author} {\bibinfo {author} {\bibfnamefont {W.}~\bibnamefont {E}},
  \bibinfo {author} {\bibfnamefont {W.}~\bibnamefont {Ren}},\ and\ \bibinfo
  {author} {\bibfnamefont {E.}~\bibnamefont {Vanden-Eijnden}},\ }\bibfield
  {title} {\bibinfo {title} {Simplified and improved string method for
  computing the minimum energy paths in barrier-crossing events},\ }\href
  {https://doi.org/10.1063/1.2720838} {\bibfield  {journal} {\bibinfo
  {journal} {The Journal of Chemical Physics}\ }\textbf {\bibinfo {volume}
  {126}},\ \bibinfo {pages} {164103} (\bibinfo {year} {2007})},\ \Eprint
  {https://arxiv.org/abs/https://doi.org/10.1063/1.2720838}
  {https://doi.org/10.1063/1.2720838} \BibitemShut {NoStop}%
\bibitem [{\citenamefont {E}\ \emph {et~al.}(2004)\citenamefont {E},
  \citenamefont {Ren},\ and\ \citenamefont
  {{Vanden-Eijnden}}}]{weinan2004minimum}%
  \BibitemOpen
  \bibfield  {author} {\bibinfo {author} {\bibfnamefont {W.}~\bibnamefont {E}},
  \bibinfo {author} {\bibfnamefont {W.}~\bibnamefont {Ren}},\ and\ \bibinfo
  {author} {\bibfnamefont {E.}~\bibnamefont {{Vanden-Eijnden}}},\ }\bibfield
  {title} {\bibinfo {title} {Minimum action method for the study of rare
  events},\ }\href@noop {} {\bibfield  {journal} {\bibinfo  {journal}
  {Communications on pure and applied mathematics}\ }\textbf {\bibinfo {volume}
  {57}},\ \bibinfo {pages} {637} (\bibinfo {year} {2004})}\BibitemShut
  {NoStop}%
\bibitem [{\citenamefont {Heymann}\ and\ \citenamefont
  {{Vanden-Eijnden}}(2008)}]{heymann2008geometric}%
  \BibitemOpen
  \bibfield  {author} {\bibinfo {author} {\bibfnamefont {M.}~\bibnamefont
  {Heymann}}\ and\ \bibinfo {author} {\bibfnamefont {E.}~\bibnamefont
  {{Vanden-Eijnden}}},\ }\bibfield  {title} {\bibinfo {title} {The geometric
  minimum action method: {{A}} least action principle on the space of curves},\
  }\href@noop {} {\bibfield  {journal} {\bibinfo  {journal} {Communications on
  Pure and Applied Mathematics: A Journal Issued by the Courant Institute of
  Mathematical Sciences}\ }\textbf {\bibinfo {volume} {61}},\ \bibinfo {pages}
  {1052} (\bibinfo {year} {2008})}\BibitemShut {NoStop}%
\bibitem [{\citenamefont {Kearnes}\ \emph {et~al.}(2016)\citenamefont
  {Kearnes}, \citenamefont {McCloskey}, \citenamefont {Berndl}, \citenamefont
  {Pande},\ and\ \citenamefont {Riley}}]{kearnes_molecular_2016}%
  \BibitemOpen
  \bibfield  {author} {\bibinfo {author} {\bibfnamefont {S.}~\bibnamefont
  {Kearnes}}, \bibinfo {author} {\bibfnamefont {K.}~\bibnamefont {McCloskey}},
  \bibinfo {author} {\bibfnamefont {M.}~\bibnamefont {Berndl}}, \bibinfo
  {author} {\bibfnamefont {V.}~\bibnamefont {Pande}},\ and\ \bibinfo {author}
  {\bibfnamefont {P.}~\bibnamefont {Riley}},\ }\bibfield  {title} {\bibinfo
  {title} {Molecular graph convolutions: Moving beyond fingerprints},\ }\href
  {https://doi.org/10.1007/s10822-016-9938-8} {\bibfield  {journal} {\bibinfo
  {journal} {Journal of Computer-Aided Molecular Design}\ }\textbf {\bibinfo
  {volume} {30}},\ \bibinfo {pages} {595} (\bibinfo {year} {2016})}\BibitemShut
  {NoStop}%
\end{thebibliography}%

\newpage
\appendix
\numberwithin{equation}{section}
\numberwithin{theorem}{section}

\section{One-dimensional example}
\label{sec:one:d}
To illustrate the necessity of importance sampling for objectives dominated by rare events, consider the one-dimensional committor problem associated with transitions between the minima located at $x=x_1$ and $x=x_2$ of  the potential $V(x)= (1-x^2)+x/10$, i.e. the minimization of
\begin{equation}
    \label{eq:com:1d}
    \int_{-1}^1 |q'(x)|^2 e^{-\beta V(x)} dx
\end{equation}
The minimizer of this objective function subject to $q(x_1)=0$, $q(x_2)=1$ is
\begin{equation}
    \label{eq:min:com:1d}
    q(x) = \frac{\int_{x_1}^x e^{\beta V(y)} dy}{\int_{x_1}^{x_2} e^{\beta V(y)} dy}
\end{equation}

\begin{figure}[h]
 \centering
 \includegraphics[width=0.48\linewidth]{1d_example_pop_loss_mean.pdf}
 \includegraphics[width=0.48\linewidth]{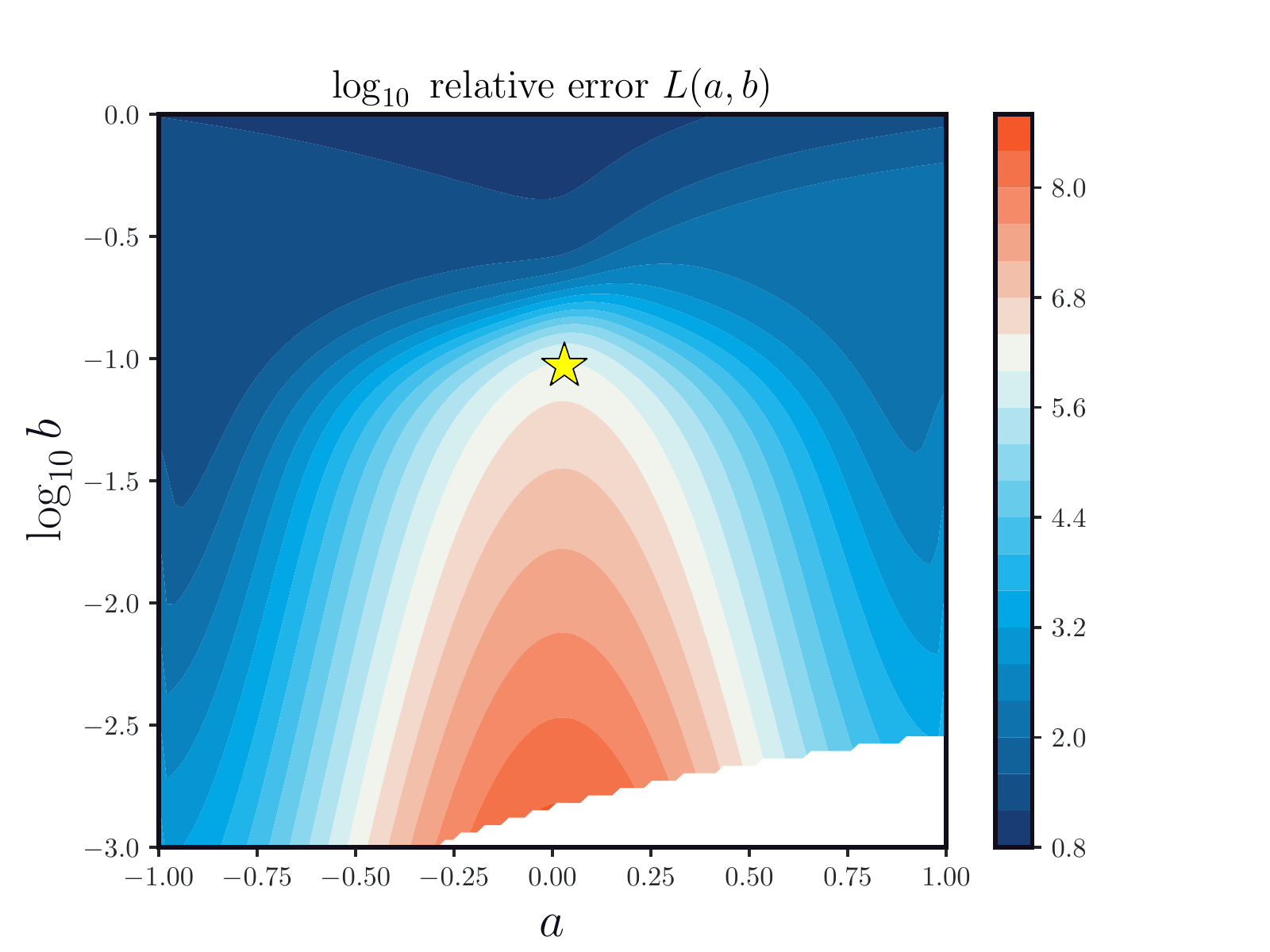}
 \caption{The loss landscape $L(a,b)$ in~\eqref{eq:loss:a:y} and the relative error on the estimator ( = std/loss) when the data is drawn from the Gibbs distribution with density $Z^{-1}_{[x_1,x_2]} = e^{-\beta V(x)}$ restricted on $x\in[x_1,x_2]$. At the minimum of the loss (located at the red dot), this relative error is  about $60$. Here $\beta = 1/8$ (i.e. the energy barrier is $8 k_BT$) and the optimal parameters are $a\approx0.04$ and $b\approx 0.11$.}
 \label{fig:loss1d}
\end{figure}

For large $\beta$, this function is sigmoid-like with a sharp transition from 0 to 1 around $x=0.1$. 
Suppose that we want to approximate it using the parametric representation 
\begin{equation}
\label{eq:qab}
    q(x;a,y) = \sigma((x-a)/b) \qquad \text{where} \qquad \sigma(z) = \frac1{1+e^{-z}}
\end{equation}
This function does not satisfy the boundary condition exactly, but for $a$ around 0.1 and $b$ small enough, it does a good job at representing the exact~\eqref{eq:min:com:1d} (see the top left panel in Fig.~\ref{fig:losscomm1d}). Accordingly, let us look at the loss function as a function of $(a,b)$ in this parameter range, viewed as an expectation of $|q'(x;a,b)|$ the Gibbs distribution with density $e^{-\beta V(x)}$ restricted to $x\in[x_1,x_2]$ and properly normalized on that interval:
\begin{equation}
    \label{eq:loss:a:y}
    L(a,b) = Z_{[-x_1,x_2]}^{-1} \int_{x_1}^{x_2} |q'(x;a,b)|^2 e^{-\beta V(x)} dx \qquad \text{with} \qquad 
    Z_{[-x_1,x_2]}^= \int_{x_1}^{x_2}  e^{-\beta V(x)} dx
\end{equation}
where
\begin{equation}
    |q'(x;a,b)|^2 = b^{-2} \sigma^2((x-a)/b)(1-\sigma((x-a)/b))
\end{equation}
The population and empirical losses were shown in in the bottom panels of Fig.~\ref{fig:losscomm1d}: the latter was obtained by drawing $10^4$ independent samples from $Z_{[-x_1,x_2]}^{-1} e^{-\beta V(x)} $ using a rejection method, resulting in the empirical distribution shown in the top right panel of Fig.~\ref{fig:losscomm1d}.
Here we compute an additional quantity:  the variance of the estimator for the population loss if we use data sampled from $Z_{[x_1,x_2]}^{-1}  e^{-\beta V(x)}$. 
The result (together with the population loss) is shown in Fig.~\ref{fig:loss1d}: when $beta$ is large so that the energy barrier is also large in units of $k_BT$ (here $\beta = 1/8$, so that the barrier is $8 k_BT$), the relative error on the loss becomes large in the regions close to the minimum of this loss. 

Note that in this one-dimensional example, adding a regularizing term to the empirical loss improves its predictions. However this strategy will not be generically applicable to higher dimensional situations.

\section{Variance reduction improves generalization}
\label{app:generr}

\begin{proof}[Proof of Proposition~\ref{th:1a}]
Recall that the discrete time updates of the stochastic gradient descent dynamics are obtained from:
\begin{equation}
    \label{eq:appgd}
    \thetab^{k+1} = \thetab^k - \alpha \nabla_{\thetab} L_n(\thetab^k), \qquad k=0,1, 2, \ldots
\end{equation}
Since the miniibatches are draw independently at every step and $\nabla_{\thetab} L_n(\thetab)$ is an unbiased estimator of $\nabla_{\thetab} L(\thetab)$, in law~\eqref{eq:appgd} is equivalent to 
\begin{equation}
    \label{eq:appsgdlaw}
        \thetab^{k+1} = \thetab^k - \alpha \nabla_{\thetab} L(\thetab^k) + \frac{\alpha}{\sqrt{n}} \nabla_{\thetab}\eta(\thetab^k), \qquad k=0,1, 2, \ldots
\end{equation}
where $\eta$ is a random function with mean zero, $\EE_\nu \eta(\thetab) = 0$, and covariance
\begin{equation}
    \EE_\nu \eta(\thetab)\eta(\thetab') = \EE_\nu \ell(\xb,\thetab)\ell(\xb,\thetab') - L(\thetab)L(\thetab').
\end{equation}
Let us introduce $\tilde \thetab_n^k$ defined as
\begin{equation}
    \tilde \thetab_n^k = \sqrt{\frac{n}{\alpha}} \big(\thetab^k - \bar \thetab^k\big) 
\end{equation} 
where $\{\bar\thetab^k\}_{k\in\NN_0}$ are the update from the GD scheme in~\eqref{eq:gd} so that 
\begin{equation}
    \label{eq:appsgdlaw12}
        \tilde \thetab_n^{k+1} = \tilde \thetab^k - \sqrt{\alpha n}   \big(  \nabla_{\thetab} L(\bar\thetab^k +\sqrt{\alpha/n} \tilde\thetab_n^k) -  \nabla_{\thetab} L(\bar\thetab^k)\big) + \sqrt{\alpha} \nabla_{\thetab}\eta(\bar\thetab^k +\sqrt{\alpha/n} \tilde\thetab_n^k).
\end{equation}
for $k=0,1,2,\ldots$.
Taking the limit as $n\to\infty$  shows that for each $k$ $\tilde \thetab_n^{k} to \tilde \thetab^{k}$, where $\{\thetab^{k}\}_{k\in\NN_0}$ is the solution of the updating scheme
\begin{equation}
    \label{eq:appsgdlaw2}
        \tilde \thetab^{k+1} = \tilde \thetab^k - \alpha H^k \tilde\thetab^k + \sqrt{\alpha} \bb^k, \qquad k=0,1, 2, \ldots
\end{equation}
where $H^k = \nabla_{\thetab} \nabla_{\thetab} L(\bar\thetab^k)$ and $\{\bb^k\}_{k\in\NN_0}$ are random vector, independent for different $k$, with mean zero and covariance 
\begin{equation}
    B^k = \EE_\nu \bb^k (\bb^k)^T = \EE_\nu \grad_{\thetab} \ell(\xb,\bar\thetab^k) (\grad_{\thetab} \ell(\xb,\bar\thetab^k))^T - \nabla_{\thetab} L(\bar\thetab^k)(\nabla_{\thetab} L(\bar\thetab^k))^T
\end{equation}
which we assume to be non-zero when the data set is finite. 

Next note that the limiting sequence $\{\tilde\thetab^k\}_{k\in\NN_0}$ can be used to deduce that
\begin{equation}
\label{eq:lim:n}
     \lim_{n\to\infty} n \EE_D [L(\thetab^k) - L(\bar \thetab^k)] = \tfrac12 \tr [ C^k H^k] \qquad \text{where} \quad C^k =\EE_D \tilde\thetab^k  (\tilde\thetab^k)^T 
\end{equation}
From~\eqref{eq:appsgdlaw2}, the tensor $C^k$ satisfies \begin{equation}
    C^{k+1} = C^k -\alpha H^k C^k - \alpha C^k H^k + \alpha^2 H^k C^k H^k + \alpha B^k
\end{equation}
with $C^0=0$ which follows from $\tilde\thetab^0 = 0$ since $\thetab^0 = \bar\thetab^0$. 
By Assumption~\ref{as:1},  as $k\to\infty$, $H^k\to H^*$, which is the positive-definite tensor defined in~\eqref{eq:min:prop}, and $B^k\to B^*$, which is the tensor defined in~\eqref{eq:B}. This guarantees that $\lim_{k\to\infty} C^k = C^*$, where $C^*$ is the solution to~\eqref{eq:C}. From~\eqref{eq:lim:n}, it also implies that
\begin{equation}
\label{eq:lim:n:k}
     \lim_{k\to\infty}\lim_{n\to\infty} n \EE_D [L(\thetab^k) - L(\bar \thetab^k)] =\lim_{k\to\infty}\lim_{n\to\infty} n \EE_D [L(\thetab^k) - L(\thetab^*)] = \tfrac12 \tr [ C^* H^*]
\end{equation}
which establishes~\eqref{eq:generr1} and ends the proof.
\end{proof}

Note that, from~\eqref{eq:appsgdlaw2}, the $k$th iterate of $\tilde\thetab^k$ is (using $\tilde\thetab^0 = 0$ with follows from $\thetab^0 = \bar\thetab^0$) 
\begin{equation}
    \tilde \thetab^k = \sqrt{\alpha} \sum_{p=0}^{k-1} \prod_{q=0}^p (1-\alpha H^q) \bb^{k-p}
\end{equation}
from which we can get more detailed information about the statistics of the sequence. Note also that, in the limit as $\alpha\to0$, \eqref{eq:appsgdlaw2} reduces  to an SDE similar to that of an Ornstein-Uhlenbeck process.

\section{Active sampling by reweighting}
\label{sec:active1}

The results of Sec.~\ref{sec:online:gen} indicate that the variance of the estimator for the gradient of the population loss dominates the generalization error. In view of this, at every step of SGD, instead of sampling the original measure $\nu$, an option is to sample a modified measure $\tilde \nu$ and reweight the samples in the estimator accordingly, in such a way as to minimize the variance of this estimator. To make this concrete let $g(\xb)=d\tilde \nu/d\nu$ be the Radon-Nikodym derivative of $\tilde\nu$ with respect to $\nu$, assume that $g$ is positive everywhere, and let $\{\tilde \xb_i\}_{i=1}^n$ be a batch of independent samples draw from $\tilde \nu$. Then 
\begin{equation}
    \label{eq:estim}
    \frac1n \sum_{i=1}^n \nabla_{\thetab} \ell(\tilde \xb_i, \thetab) g^{-1}(\xb_i)
\end{equation}
is an unbiased estimator for the gradient of population loss and the choice of $\tilde \nu$ that minimizes the variance of this estimator, i.e. minimizes
\begin{equation}
    \label{eq:estim2}
    \int_\Omega  |\nabla_{\thetab} \ell(\tilde \xb, \thetab)|^2 g^{-2}(\xb) d\tilde \nu(\xb) = \int_\Omega  |\nabla_{\thetab} \ell(\tilde \xb_i, \thetab)|^2 g^{-1}(\xb_i) d\nu(\xb),
\end{equation}
is 
\begin{equation}
\label{eq:tilted:nu}    
    d\tilde \nu(\xb) = g(\xb) d\nu(\xb) \qquad \text{with} \qquad
    g(\xb) = \frac{|\nabla_{\thetab} \ell(\xb,\thetab)|}{\EE_\nu |\nabla_{\thetab} \ell(\cdot,\thetab)|}
\end{equation}
An obvious difficulty with this estimator is that the reweighting factor $g(\xb)$ contains the factor $\EE_\nu |\nabla_{\thetab} \ell(\cdot,\thetab)|$ which we do not know. Still, in the context of optimization by SGD, it is useful since any unknown constant entering the gradient of the loss can be absorbed in the learning rate. To see why consider the following scheme: Starting from some initial value $\tilde \thetab^0$, update these parameters using the iteration rule
\begin{equation}
    \label{eq:sgd:imp}
    \tilde\thetab^{k+1} = \tilde \thetab^k - \frac\alpha n \sum_{i=1}^n 
    \frac{\nabla_{\thetab}\ell(\tilde \xb_i,\tilde \thetab^k)}{ |\nabla_{\thetab}\ell(\tilde \xb_i,\tilde \thetab^k)|}, \qquad k=0,1, 2, \ldots
\end{equation}
where the batch $\{\tilde \xb_i\}_{i=1}^n$ contains independent samples from 
\begin{equation}
    \label{eq:gibbs1}
    d\tilde \nu_k(\xb) = \tilde Z_k^{-1} |\nabla_\theta \ell(\xb,\tilde\thetab^k)| d\nu(\xb) \quad \text{with} \quad \tilde Z_k = \EE_\nu |\nabla_\theta \ell(\cdot,\tilde\thetab^k)|.
\end{equation}
Note that this measure can be sampled by the Metropolis-Hastings method or the Metropolis-adjusted Langevin algorithm without requiring to know its normalization factor $\tilde Z_k$. Under  Assumption~\ref{as:1} we can prove the following equivalent of \eqref{eq:generr1} 
\begin{proposition}
\label{th:1b}
The sequence $\{\tilde \thetab^k\}_{k\in \NN_0}$ obtained using the SGD update in~\eqref{eq:sgd:imp} starting from $\tilde \thetab^0 = \bar\thetab^0$ and using an independent batch of data $\{\tilde \xb_i\}_{i=1}^n$ drawn from $\tilde \nu^k$ at every step is such that
\begin{equation}
\label{eq:generr:imp}
    \lim_{k\to\infty} \lim_{n\to\infty} n \EE_D [\tilde L_n( \thetab^k) - L(\thetab^*)] = \tfrac12\alpha \tr [\tilde C^* H^*]
\end{equation} 
where $\EE_D$ denotes expectation over all the batches used to compute the sequence $\thetab^k$, and $C^*$ is the $N\times N$ tensor that solves
\begin{equation}
    \label{eq:C:imp}
    H^* \tilde C^* + \tilde C^* H^* - \alpha \tilde C^* H^* \tilde C^* = \tilde B^*
\end{equation}
Here $\tilde B^*$ is 
\begin{equation}
    \label{eq:B:imp}
    \tilde B^* = \int_\Omega \frac{\nabla_{\thetab} \ell(\xb,\thetab^*) [\nabla_{\thetab} \ell(\xb,\thetab^*)]^T}{|\nabla_{\thetab} \ell(\xb,\thetab^*)|^2}  d\tilde \nu_*(\xb)
\end{equation}
where 
\begin{equation}
    \label{eq:gibbs2}
    d\tilde \nu_*(\xb) = \tilde Z_*^{-1} |\nabla_\theta \ell(\xb,\thetab_*)| 
    d\nu(\xb) \quad \text{with} \quad \tilde Z_* = \EE_\nu |\nabla_\theta \ell(\cdot,\tilde\thetab_*)|.
\end{equation}
\end{proposition}

\noindent The proof of this proposition is similar to that of Proposition~\ref{th:1a}. For small $\alpha$, this shows again that the error will be controlled by $\tr \tilde B^*$, which is now trivially given by
\begin{equation}
    \label{eq:tr:B:imp}
    \tr \tilde B^* = 1
\end{equation}
This result may look surprising but it is a consequence of the fact that, by using~\eqref{eq:sgd:imp}   we have effectively absorbed in the learning rate the unknown factor $\EE_\nu |\nabla_{\thetab} \ell(\cdot,\thetab)|$ entering the weights $g(\xb)$ defined in~\eqref{eq:tilted:nu}. If we had not done this, $\tr \tilde B^*$ in~\eqref{eq:tr:B:imp} would be replaced by $|\EE_\nu |\nabla_{\thetab} \ell(\cdot,\thetab)||^2$; this provides a point of comparison with the scheme discussed in Proposition~\ref{th:1a}, since from~\eqref{eq:B} $\tr B^* = \EE_\nu |\nabla_{\thetab} \ell(\cdot,\thetab)|^2 \ge |\EE_\nu |\nabla_{\thetab} \ell(\cdot,\thetab)||^2$. Therefore we would reduce the variance. 

Coming back to the scheme defined by~\eqref{eq:sgd:imp}, one feature that makes it somewhat academic is that we still need to sample $\tilde \nu_k$: while this can in principle be done via the Metropolis-Hastings method or the Metropolis-adjusted Langevin algorithm, we have no guarantees that this sampling will be fast---for example, even if $\tilde \nu $ has a density $\rho(\xb)$ with respect to the Hausdorff measure on $\Omega$, there is no guarantee that its potential $-\log \rho(\xb)$ will be convex or even that it will have a single minimum. For these reasons, we instead implement the alternative active importance sampling strategy based on umbrella sampling and replica exchange which we deem more robust and more widely applicable.  

\section{Approximation of the committor with a neural network}
\label{app:rep}
\subsection{Representation}

 Neural networks (NN) offer flexibility to the representation and relative ease of optimization, making them a natural choice for a representation of the committor. For example, if we use a single hidden layer neural network with nonlinearity $\varphi$ (e.g., ReLU) passed through a thresholding function $\sigma$ (e.g., a sigmoid function, $\sigma(z) = 1/(1+e^{-z})$) to ensures that $q(\xb) \in [0,1],\ \forall \xb \in \RR^d$, this amounts to taking
\begin{equation}
  \label{eq:nn}
 q(\xb,\thetab) = \sigma \left[ \frac1n \sum_{i=1}^n \varphi(\xb, \thetab_i) \right]
\end{equation}
where we use $\thetab_i$ with $i=1,\ldots,n$ to denote the parameters in each neural units and $\thetab= (\thetab_1,\ldots,\thetab_n)$ to denote all of them collectively. In practice, the architecture of the neural network will be substantially more intricate than the single hidden layer network~\eqref{eq:nn}.

% In the optimization procedure below, it is more tractable to penalize
% deviations from the boundary conditions rather than impose them as
% constraints.  Consequently, we use a strong Lagrange multipliers to
% ensure that the committor has the right values on the initial and
% target states.  The objective function we use is thus
% \begin{equation}
%  \label{eq:appobjfun1}
%  \begin{aligned}
%   C_\lambda[q] & = Z^{-1} \int_{\RR^d} |\grad q(\xb)|^2 e^{-\beta
%      V(\xb)} d\xb \\
%   & + \lambda Z^{-1} \int_A q(\xb)^2 e^{-\beta V(\xb)} d\xb+ \lambda
%   Z^{-1} \int_B \bigl( 1-q(\xb) \bigr)^2 e^{-\beta V(\xb)} d\xb\\
%   & \equiv \left\<|\grad q|^2 \right\>_\beta +\lambda \left\<
%      |q|^2 1_A\right \>_\beta+ \lambda \left\<
%      |1-q|^2 1_B\right \>_\beta
%  \end{aligned}
% \end{equation}
% where $\<\cdot\>_\beta$ denotes canonical expectation with respect to
% $Z^{-1} e^{-\beta V(\xb)} $, and $1_A$ and $1_B$ are the indicator
% functions of $A$ and $B$, respectively.  Using the parametric
% representation of the committor in \eqref{eq:nn} the problem becomes
% to find a set of $\{ \thetab \}_{i=1}^n$ that minimize $C_\lambda$.

\subsection{Computing the gradients}

Optimization of the neural network representation of the committor~\eqref{eq:nn} by gradient descent (GD) requires estimating the gradient of the objective function with respect to the parameters. For example, if we use~\eqref{eq:nn} in the Lagrangian defined in~\eqref{eq:lag:com}, we have
\begin{equation}
\label{eq:grad:lag:com}
   \tfrac12 \grad_{\thetab_i} \mathcal{L}(\xb,q)  \equiv\tfrac12 \grad_{\thetab_i} \ell(\xb,\thetab) 
   = \grad_{\thetab_i} \grad_{\xb}  q \grad_{\xb}  q  +\lambda 
     q \grad_{\thetab_i}  q 1_A- \lambda 
     (1-q) \grad_{\thetab_i}  q 1_B 
\end{equation}
Noting that, with $\sigma(z) = 1/(1+e^{-z})$,
\begin{equation}
 \grad_{\xb} \sigma(f(\xb)) = \sigma(f(\xb))\left( 1-\sigma(f(\xb)) \right) \grad_{\xb} f(\xb)
\end{equation}
and similarly for $\grad_{\thetab}$ we can derive explicit expressions for the factors at the right hand side of~\eqref{eq:grad:lag:com}.  In particular, we see that
\begin{equation}
  \grad_{\xb} q(\xb,\thetab) =
  \frac1n q(\xb,\thetab)(1-q(\xb,\thetab))\sum_{i=1}^n \grad_{\xb}\phi(\xb)
  \grad_\phi\varphi(\phi(\xb),\thetab_i),
\end{equation}
\begin{equation}
 \grad_{\thetab_i} q(\xb,\thetab) = \frac1n
 q(\xb,\thetab)(1-q(\xb,\thetab)) \grad_{\thetab_i}\varphi(\phi(\xb),\thetab_i)
\end{equation}
and
\begin{equation}
  \begin{aligned}
    \grad_{\thetab_i} \grad_{\xb} q(\xb,\thetab) & = \frac1n
    \grad_{\thetab_i} q(\xb,\thetab)(1-2q(\xb,\thetab))\sum_{j=1}^n
    \grad_{\xb}\phi(\xb)
    \grad_\phi\varphi(\phi(\xb),\thetab_j)\\
    & + \frac1n q(\xb,\thetab)(1-q(\xb,\thetab))\grad_{\xb}\phi(\xb)
  \grad_\phi \grad_{\thetab_i}\varphi(\phi(\xb),\thetab_i).
  \end{aligned}
\end{equation}

\section{Alternative formulation of the committor and boundary conditions}
\label{app:charges}

The variational problem of determining the committor function can be reinterpreted via a solution to the following PDE~\cite{lu_exact_2014},
\begin{equation}
  L\tilde q = \tau e^{\beta V(\xb)} \left[ \delta(\xb-\ab) - \delta(\xb-\bb) \right].
  \label{eq:charges}
\end{equation}
where $\tau>0$ is a parameter, and $\delta(\xb-\ab)$ and $\delta(\xb-\bb)$ denote the Dirac delta distribution centered at $\ab$ and $\bb$ respectively.
Given a solution to \eqref{eq:charges}, it is straightforward to verify that the committor between the sets
\begin{equation}
  \nonumber
  \begin{aligned}
  A = \{ \xb | \tilde q(\xb) \leq \tilde q_- \} \ni \ab\\
  B = \{ \xb | \tilde q(\xb) \geq \tilde q_+ \}\ni \bb
\end{aligned}
\end{equation}
is given by
\begin{equation}
  q(\xb) = \frac{\tilde q(\xb) - \tilde q_-}{\tilde q_+ - \tilde q_-}
\end{equation}
for $\xb\in (A\cup B)^c$.

We can use the variational optimization algorithm Alg.\ref{alg:basic} to compute $\tilde q$ where we penalize the cost functional to obtain the loss function,
\begin{equation}
  C_{\lambda}[\tilde q] = C[\tilde q] + \tau  \left(\tilde q(\ab)  - \tilde q(\bb)\right) .
\end{equation}
This formulation offers several advantages compared to the formulation discussed in the main text.
First, because the range of $\tilde q$ is all of $\RR$, there is no need to use thresholding functions that could affect the magnitude of gradients and hence the rate of convergence of the optimization.
Secondly, to use the penalized objective of the main text, we must draw samples from the metastable states $A$ and $B$.
If those states are difficult to sample, the boundary conditions here require knowledge of only two points $\ab\in A$  and $\bb\in B$.

\end{document}